%% file: main.tex
\documentclass[a4paper,11pt]{article}
\pdfoutput=1

\usepackage{xcolor}
\usepackage{jheppub}
\usepackage[T1]{fontenc}
\usepackage{mathtools,bm,amsthm,microtype}

\hypersetup{
  colorlinks=true,
  linkcolor=red!45!black,
  citecolor=blue!55!black,
  urlcolor=blue!70!black
}

\DeclareMathOperator{\Tr}{Tr}
\DeclareMathOperator{\Ran}{Ran}
\newcommand{\cA}{\mathcal A}
\newcommand{\cD}{\mathcal D}
\newcommand{\cH}{\mathcal H}
\newcommand{\Id}{I}
\newcommand{\ket}[1]{\lvert #1\rangle}
\newcommand{\bra}[1]{\langle #1\rvert}

\newtheorem{theorem}{Theorem}[section]
\newtheorem{lemma}[theorem]{Lemma}
\newtheorem{proposition}[theorem]{Proposition}

\title{One cut, two parities:\\
exact Fibonacci charge responses in the golden chain}
\author{Yi Liang}
\affiliation{Independent Researcher, Beijing, China}
\emailAdd{iantrans2042@gmail.com}

\abstract{We determine the exact cut-charge Born law in the ground state of every
periodic antiferromagnetic golden chain with at least three sites.
Its two parity branches retain different information.  At even length, the
ratio of nontrivial to vacuum probabilities equals the squared Fibonacci
quantum dimension.  The odd vacuum probability equals the even-branch
nontrivial probability times one ground-state Born coordinate, while
normalization fixes the complement.  Parity selects the positive hoop sector
at even length and the complementary sector at odd length.  At each length,
the corresponding cut projector compresses to a scalar in the even
sector or to a non-scalar element of a two-character algebra in the odd
sector.  The unchanged measurement rule yields categorical rigidity at even
length and one-coordinate state sensitivity at odd length.}

\begin{document}
\maketitle

\input{sections/introduction}
\input{sections/model-and-result}
\input{sections/parity-mechanism}
\input{sections/physical-response}
\input{sections/odd-coordinate}
\input{sections/exact-odd-evaluation}
\input{sections/discussion}

\appendix
\input{appendices/projection-geometry}
\input{appendices/carrier-covariance}
\input{appendices/perron-response}
\input{appendices/odd-packet}
\input{appendices/pointed-height}
\input{appendices/odd-evaluator}

\section*{Code availability}
A standalone Python package reproducing the finite-size cut-charge
calculations and the spectral-adjugate evaluation reported in this work is
openly archived on Zenodo as version 2.0.0 \cite{ZenodoCode}.  The package
includes JSON reference outputs for direct comparison.

\bibliographystyle{JHEP}
\bibliography{refs}
\end{document}

%% file: sections/introduction.tex
\section{Introduction}
\label{sec:introduction}

Non-invertible symmetry broadens the organizing principles of quantum
many-body physics beyond the familiar group-based framework
\cite{BhardwajTachikawa2018,Shao2023TASI,SchaferNameki2024}.  Yet categorical organization becomes physically
incisive only when it constrains what a specified measurement can reveal.  For
a fixed quantum state and measurement, which features of the resulting
probability distribution are fixed by categorical structure, and which remain
encoded in the state itself?  This question asks for the operational content
of non-invertible symmetry.  An exact answer would separate information fixed
by the symmetry structure from state-dependent information, clarify what the
symmetry predicts for an observable response, and show how an abstract
organization of sectors becomes a concrete statement about a many-body
system.

Two developments have made that question pressing.  Structural work shows that
topological defect lines realize categorical fusion, diagnose anomalies, and,
when preserved along renormalization-group flows, can restrict the possible
gapped infrared endpoints \cite{ChangEtAl2019}.  Operational work has
formulated local non-invertible actions as completely positive quantum
operations, implemented lattice Kramers--Wannier duality while tracking charge
and twisted-sector data, and realized fusion-category symmetries through
probability-preserving isometries and trace-preserving channels between twisted
sectors \cite{OkadaTachikawa2024,KhanKhanMohd2024,
BartschGaiSchaferNameki2026}.  Together, these advances move the field from
structural constraints toward an explicit operational language of probability,
distinguishing the input state, applied operation, outgoing sector, and
readout.  A fixed-state, fixed-measurement Born response follows a different
data flow.  It starts from a specified state within a fixed many-body
ground-state family, leaves the observable unchanged, and asks for its outcome
distribution.  Outgoing-channel probabilities and fixed-observable Born
statistics are distinct objects.  The unresolved interface is how that
distribution divides between structure-fixed and state-dependent information.
Resolving it requires an exact finite-volume setting in which state, observable,
and categorical dictionary are simultaneously controlled.

The golden chain supplies this exact finite-volume setting.  Restricted
solid-on-solid models organize integrable face weights through local-height
adjacency \cite{AndrewsBaxterForrester1984}.
In its antiferromagnetic regime, the golden chain is related to the critical
\(A_4\) RSOS model, whose scaling theory is the tricritical Ising conformal
field theory with central charge \(7/10\)
\cite{Feiguin2007golden,Trebst2008introduction,Pfeifer2012translation}.
The \(A_4\)-to-Fibonacci fold maps the \(A_4\)
height graph to the Fibonacci fusion-path adjacency and thereby provides a
finite-lattice dictionary between folded height data and anyon path labels.
On the ring, the affine Temperley--Lieb algebra and topological hoop operators
provide the corresponding finite-lattice sector structure
\cite{Belletete2020fusion,AasenFendleyMong2020}.  In the same exact-lattice lineage, selected
one-site reduced-density-matrix entries of the face models have been evaluated
as local-height probabilities in closed form
\cite{FrahmWesterfeld2021,WesterfeldEtAl2023}.  Coupled anyon chains with
non-invertible symmetries provide a contemporary neighboring setting
\cite{AntunesRong2026CoupledAnyonChains}.  These ingredients together
make the lineage suited to the fixed-readout question, combining an
interacting finite-size ground-state family, a controlled algebraic dictionary
and exact methods for local-height probabilities.

We therefore fix one physical family.  For every finite \(L\geq3\), take the
periodic antiferromagnetic golden chain on the cyclic Fibonacci fusion-path
Hilbert space \cite{Feiguin2007golden}, with the same local Hamiltonian
prescription at every length.  Its ground ray is unique, and we write
\(\Omega_L\) for the normalized representative that is strictly positive in
the fusion-path basis.  Lemma~\ref{lem:positive-ground} establishes uniqueness
and strict positivity.  Even and odd lengths form two towers within this
family rather than two models.  The periodic boundary condition, charge
alphabet \(\{\mathbf 1,\tau\}\) and ground-state selection rule are held fixed;
only the site number \(L\) changes.  This family of finite periodic lattices
is the physical regime addressed here.

One cut projector completes the setting.  On each ring choose a fusion edge
\(\ell\) as the cut, and let \(\Pi_{a,L}^{(\ell)}\) be the orthogonal projector onto
the fusion-path label \(a\in\{\mathbf 1,\tau\}\) there.  It is distinct both
from the global topological hoop operator associated with periodic geometry
and from the local two-anyon projector entering the Hamiltonian.  For a
connected interval \(A\) bounded by \(\ell\) and a second edge \(\ell'\),
order its joint endpoint weights \(p_{a,b}\) so that \(a\) labels \(\ell\)
and \(b\) labels \(\ell'\).  The anyonic Schmidt decomposition is graded by
this pair \cite{Pfeifer2014anyonicEntanglement,KatoFurrerMurao2014,
BondersonKnappPatel2017}.  The fusion-edge marginal probability
\(P_a^{(\ell)}=\langle\Omega_L,\Pi_{a,L}^{(\ell)}\Omega_L\rangle
=\sum_{b\in\{\mathbf 1,\tau\}}p_{a,b}\) leaves \(\ell'\) unobserved and is
therefore a one-endpoint marginal rather than the complete two-endpoint
distribution.  Although \(\Pi_{a,L}^{(\ell)}\) is diagonal at one path bond,
its label records accumulated fusion charge.  Path-diagonal locality therefore
differs from microscopic locality among the physical anyons.  Recent circuit,
hardware, and measurement-only proposals also treat fusion or anyon-charge
readout as an operational primitive
\cite{Ayeni2026NativeTopologicalReadout,BylesHornerVarcoePachos2026,
MinevEtAl2025StringNetCondensation,BseisoEtAl2024MinimalFibonacciCircuits,
Bonderson2021MeasuringTopologicalOrder}.  Those settings differ in state
preparation, geometry, and output.  The same
mathematical bond-label projector appears in an open constrained-chain
ensemble \cite{ChandranSchulzBurnell2016}.  The overlap is at the operator
level, while the present Born response is defined by the periodic geometry and
antiferromagnetic ground state.

What remains missing is a finite-volume account of how parity changes the
information carried by one fixed measurement.  We pose the question for the
periodic antiferromagnetic ground-state family at every \(L\geq3\), with the
Hamiltonian prescription, the charge alphabet \(\{\mathbf 1,\tau\}\) and the
cut all held fixed.  Any complete answer must link the
parity-selected hoop sector to scalar or two-character compression of the same
projector, state both finite-size Born branches, and evaluate the ground-state
coordinate that survives at odd length.  How can parity alone alter what this
fixed measurement retains about the corresponding ground state?

The answer is exact at each finite length and covers both parity towers.  With
\(\varphi=(1+\sqrt5)/2\), set \(\delta=(1+\varphi^2)^{-1}\) and
\(c=\varphi^2/(1+\varphi^2)\).  Then, for every \(L\geq3\) and every choice
of cut,
\[
  (P_{\mathbf 1},P_\tau)=
  \begin{cases}
    (\delta,c),&L\ \mathrm{even},\\
    (cq_L,1-cq_L),&L\ \mathrm{odd}.
  \end{cases}
\]
For odd \(L\),
\(q_L=\langle\Omega_L,B_L\Omega_L\rangle\in[0,1]\) is the Born coordinate of
the normalized \(L\)-site ground state in the cut-active character projector
\(B_L\).  Each line is evaluated in its corresponding \(L\)-indexed ground
state, and together they form the exact branches of one cut-charge response
law.  Categorical data fix the even pair outright; at odd length they fix the
dependence on one state-resolved coordinate, whose value is supplied by
\(\Omega_L\).  The even pair coincides with the critical \(A_4\) one-site
local-height values reported in
\cite{FrahmWesterfeld2021,WesterfeldEtAl2023} after the
\(A_4\)-to-Fibonacci fold of \cite{Belletete2020fusion}.  The complete
all-\(L\) statement is theorem~\ref{thm:parity-cut-charge}.

The finite-lattice hoop and defect setting comes from
\cite{Belletete2020fusion,BelleteteEtAl2023}.  Within it, the mechanism proved
here has two exact steps.  First, in the fixed real fusion-path gauge, every
nonzero entry of the transported \(\tau\)-line reference column has sign
\((-1)^L\).  Combined with uniqueness and strict positivity of \(\Omega_L\),
this mathematical column identity gives a strictly positive test vector in
the parity-compatible hoop sector.  Its nonzero overlap with the ground ray,
together with sector orthogonality, selects the positive hoop sector at even
length and the complementary sector at odd length.  Second, compression of
the unchanged cut projector into the selected sector is scalar at even
length and resolves into two nonzero character projectors at odd length.
Scalar compression fixes both even probabilities of this binary readout,
whereas odd compression leaves the single normalized ground-state coordinate
\(q_L\) visible.  Parity therefore changes the response's state dependence
while the Hamiltonian prescription, charge alphabet and cut measurement rule remain
fixed.

A common construction realizes both branches.  In the face-model and
lattice-defect lineage
\cite{AndrewsBaxterForrester1984,Belletete2020fusion}, an \(A_4\)
representation covering both parities intertwines the local generators and
hence the Hamiltonian, maps \(\Omega_L\) to its Perron row state, and transports the cut
projector under one observable dictionary.  A positive congruence marks that
transported projector and yields a simple Perron root whose logarithmic
derivative at zero returns the original cut expectation on either parity tower
(theorem~\ref{thm:marked-perron-response}).  At odd length, pairing the
ground-state spectral adjugate with the cut projector and dividing by
\(c\) gives \(q_L\) as a ratio of finite matrix traces
\eqref{eq:physical-odd-adjugate-evaluation}.  This exact finite-dimensional
representation supplies a calculational route to the surviving coordinate.
Together, the \(A_4\) representation and evaluator support the law and prepare the direct
comparisons that follow, one through the even formula and one through the
projector.

Two direct antecedents meet the present result along different axes.  Along the
formula axis, critical \(A_4\) one-site local-height probabilities reported in
the RSOS/RDM lineage fold to the same even pair.  Those calculations use
thermodynamic ground-state input
\cite{FrahmWesterfeld2021,WesterfeldEtAl2023}.  Along the observable axis, the
same mathematical fusion-path bond-label projector was defined in an open
constrained Fibonacci chain, where exact finite-\(N\) infinite-temperature
probabilities were derived and mid-spectrum thermalization behavior was studied
\cite{ChandranSchulzBurnell2016}.  These establish formula and projector
overlaps, respectively.  Neither entails identity of the resulting physical
probabilities, which also requires the state, geometry and regime of the
expectation value to agree.  In the present work, the unchanged readout is
evaluated in the corresponding unique periodic antiferromagnetic ground state
on a closed ring for every finite length.  Both parity branches enter the exact
law on the same footing, while parity determines the type of information
retained by the response.

The neighboring probability architectures separate when compared by state,
operation, observable, geometry, regime and output.  Categorical channels act
on twisted-sector input states by isometries or trace-preserving channels and
output outgoing-sector transition probabilities
\cite{BartschGaiSchaferNameki2026}.  Tube-sector distinguishability is a
one-shot state-discrimination problem at an entanglement cut.  Its
covariance-constrained POVM uses primitive central idempotents as coarse
outcomes and conditional information within a fiber to refine their
probabilities \cite{He2026TubeDistinguishability}.  Categorical SREE resolves a
full bipartite reduced state with boundary-condition-dependent projectors at
two entangling surfaces and outputs sector-resolved entanglement data in a
continuum or asymptotic regime.  The associated anyonic tests use open critical
chains \cite{SauraBastida2024catSREE,HeymannQuella2025}.  The present work
applies no symmetry operation and holds the periodic antiferromagnetic
ground-state prescription, a binary cut-charge readout and closed-ring geometry
fixed.  It works at every finite length and outputs the Born law of that
readout; only \(L\) varies, with its parity determining the information type.
All four programs turn categorical structure into probabilities, while the six
coordinates assign those probabilities to distinct physical objects.  This
fixes which of those objects the present probabilities describe.

The parity-complete law is exact at every finite \(L\geq3\) for a single
periodic antiferromagnetic golden-chain ground-state family and one cut
projector on each ring.  The statement fixes the state,
observable, closed-ring geometry and finite-volume regime before any scaling
question arises.  On the category-fixed even branch, the cut-charge
probabilities give
\(g_{\mathrm{cut}}:=\sqrt{P_\tau/P_{\mathbf 1}}=\varphi\), an operational
finite-lattice diagnostic.  Affleck--Ludwig boundary entropy is a continuum
boundary-state quantity \cite{AffleckLudwig1991}.  Comparing this lattice
diagnostic with boundary entropy would require a specified boundary-state
dictionary and controlled scaling geometry.  Lattice Ising defect
entanglement provides an example in which zero-mode finite-size corrections
affect comparison with continuum expectations \cite{RoySaleur2022}.  The exact
cut ratio thus supplies a concrete lattice datum for testing a future
boundary-state dictionary.

We specify the model, cut projector and endpoint dictionary in
section~\ref{sec:model-result}, where we also state the parity-complete
theorem.  Its proof through sector selection and scalar or two-character
compression occupies section~\ref{sec:parity-mechanism}.  A Perron-root
derivative common to both branches is constructed in
section~\ref{sec:physical-response}.  The contrast between categorical
rigidity and odd state sensitivity is developed in
section~\ref{sec:odd-coordinate}, followed by an exact finite-size evaluation
of the odd coordinate in section~\ref{sec:exact-odd-evaluation}.
Section~\ref{sec:discussion} compares the result with neighboring frameworks,
separates portable structure from model-specific input and outlines open
directions.  Supporting proofs are collected in the appendices.

%% file: sections/model-and-result.tex
\section{Model, cut observable, and parity-complete theorem}
\label{sec:model-result}

Before stating the theorem, we fix the periodic antiferromagnetic golden-chain
ground-state family and the cut-charge measurement on each ring.
Distinguishing the resulting one-endpoint marginal from the local Hamiltonian
projector and global hoop sectors provides the common object dictionary for
both parity branches.

\subsection{Fibonacci chain and AFM ground state}
\label{sec:model-cut}

The Fibonacci fusion category has two simple objects, \(\mathbf 1\) and
\(\tau\) \cite{Kitaev2006,NayakEtAl2008}; the fusion and quantum-dimension
data used below are
\begin{equation}
  \tau\otimes\tau=\mathbf 1\oplus\tau,
  \qquad
  d_{\mathbf 1}=1,
  \qquad
  d_\tau=\varphi,
  \qquad
  \varphi=\frac{1+\sqrt5}{2},
  \qquad
  \cD^2=1+\varphi^2.
  \label{eq:fibonacci-data}
\end{equation}
Our microscopic setting is the periodic antiferromagnetic golden chain
\cite{Feiguin2007golden,Trebst2008collective,Gils2009edge,GilsArdonne2013}.
Cyclic words label an orthonormal basis of its fusion-path Hilbert space
\(\widehat{\cH}_L\),
\begin{equation}
  x=(x_1,\ldots,x_L),
  \qquad
  x_j\in\{0,1\},
  \qquad
  0\equiv\mathbf 1,\quad 1\equiv\tau,
  \label{eq:path-basis}
\end{equation}
subject to the constraint that no two cyclically adjacent letters both equal
\(0\).  With rows and columns ordered by vacuum and then \(\tau\), the same
admissibility rule is encoded by
\begin{equation}
  N=
  \begin{pmatrix}
    0&1\\
    1&1
  \end{pmatrix}.
  \label{eq:fusion-adjacency}
\end{equation}
In the real unitary gauge, with the intermediate channels in the same order,
the only nontrivial matrix-valued associator needed below is
\begin{equation}
  F^{\tau\tau\tau}_{\tau}
  =
  \begin{pmatrix}
    \varphi^{-1}&\varphi^{-1/2}\\
    \varphi^{-1/2}&-\varphi^{-1}
  \end{pmatrix}.
  \label{eq:fibonacci-F}
\end{equation}
The periodic Temperley--Lieb generators are Hermitian and obey
\(e_j^2=\varphi e_j\), so that \(P_j=e_j/\varphi\) projects locally onto the
vacuum fusion channel.  Our normalization of the antiferromagnetic
Hamiltonian is
\begin{equation}
  H_L^{\mathrm{AFM}}=-\sum_{j=1}^{L}e_j
  \label{eq:afm-hamiltonian}
\end{equation}
and \(\Omega_L\) denotes its normalized ground vector.  Its unique positive
ray is established in the next section.

The uniform theorem is stated for \(L\geq3\).  On the two-site ring, periodic
indexing makes the two neighbours of each local generator coincide, producing
an exceptional short-ring geometry.

\subsection{Fusion-edge cut and one-endpoint marginal}
\label{sec:named-cut}

Fix one fusion edge \(\ell\) in the cyclic path and call it the cut.  The charge
observable at this edge reads the fusion-path label itself and is distinct
from the local two-anyon projector \(P_j\) entering the Hamiltonian.  Resolving
the two possible labels splits the space orthogonally,
\begin{equation}
  \widehat{\cH}_L=\cH_{\mathbf 1}^{(\ell)}
  \oplus\cH_{\tau}^{(\ell)},
  \qquad
  \Pi_L\equiv\Pi_{\mathbf 1,L}^{(\ell)},
  \qquad
  \Pi_{\tau,L}^{(\ell)}=\Id-\Pi_L.
  \label{eq:named-cut-decomposition}
\end{equation}
This direct sum resolves one fusion-path edge.  By comparison, a connected
interval \(A\) of the ring defines a spatial bipartition with two boundary
edges, say \(\ell\) and \(\ell'\).  We order the endpoint pair so that \(a\)
labels \(\ell\) and \(b\) labels \(\ell'\).  The full Schmidt space is graded
by this \emph{pair} of endpoint charges
\cite{Pfeifer2014anyonicEntanglement,KatoFurrerMurao2014,
BondersonKnappPatel2017},
\begin{equation}
  \cH_A\otimes_{\mathrm{fus}}\cH_{\bar A}
  =
  \bigoplus_{a,b\in\{\mathbf 1,\tau\}}
  \cH_A^{(a,b)}\otimes\cH_{\bar A}^{(b,a)}.
  \label{eq:fusion-schmidt}
\end{equation}
The complementary interval carries the endpoint labels in reverse order;
because the Fibonacci charges are self-dual, no conjugation changes those
labels.  With \(\rho_A\) normalized by the ordinary Hilbert-space trace in
this orthogonal Schmidt decomposition, let \(\rho_A^{(a,b)}\) denote its
corresponding block.  The reduced density matrix is block diagonal in
\((a,b)\), and the complete interval-block weights are
\(p_{a,b}:=\Tr\rho_A^{(a,b)}\).  The projector
\(\Pi_{a,L}^{(\ell)}\), by contrast, resolves only the first endpoint label
and does not condition on \(\ell'\).  Its Born probability is therefore the
one-endpoint marginal
\begin{equation}
  P_a^{(\ell)}=\sum_{b\in\{\mathbf 1,\tau\}}p_{a,b}
  =
  \langle\Omega_L,\Pi_{a,L}^{(\ell)}\Omega_L\rangle.
  \label{eq:charge-probabilities}
\end{equation}
Operationally, this marginal can be obtained either by resolving only the
endpoint charge at \(\ell\) or by resolving both endpoint charges and summing
the joint probabilities over the outcome at \(\ell'\).
This marginal, rather than the two-endpoint distribution \(\{p_{a,b}\}\), is
what ``cut-charge probability'' means everywhere in this paper.  We keep the
cut fixed and henceforth suppress the superscript \((\ell)\).  For the
vacuum outcome we write
\begin{equation}
  \chi_L:=
  \langle\Omega_L,\Pi_L\Omega_L\rangle
  =P_{\mathbf 1}
  \label{eq:vacuum-cut-probability}
\end{equation}
which is the probability of finding vacuum charge at the cut.

\subsection{Hoop sectors}
\label{sec:hoop-sectors}

Wrapping a \(\tau\) loop around the periodic chain defines the nontrivial
hoop \(Y_{\tau,L}\) \cite{Pfeifer2012sectors,BuicanGromov2017,
Belletete2020fusion,BelleteteEtAl2023}.  Self-duality makes the hoop
Hermitian, the fusion rule gives its quadratic relation, and topological
naturality makes it commute with the periodic Temperley--Lieb Hamiltonian:
\begin{equation}
  Y_{\tau,L}^{\dagger}=Y_{\tau,L},
  \qquad
  Y_{\tau,L}^2=\Id+Y_{\tau,L},
  \qquad
  [Y_{\tau,L},H_L^{\mathrm{AFM}}]=0.
  \label{eq:hoop-algebra}
\end{equation}
Only \(\varphi\) and \(-\varphi^{-1}\) can therefore occur as eigenvalues,
and the associated orthogonal projectors read
\begin{equation}
  E_{+,L}=\frac{\Id+\varphi Y_{\tau,L}}{\cD^2},
  \qquad
  E_{-,L}=\Id-E_{+,L}
  =\frac{\varphi^2\Id-\varphi Y_{\tau,L}}{\cD^2}.
  \label{eq:hoop-projectors}
\end{equation}
These projectors resolve the global topological sectors:
\(E_{+,L}\) and \(E_{-,L}\) project onto the \(\varphi\) and
\(-\varphi^{-1}\) hoop eigenspaces, respectively.  The cut projector
instead resolves a fusion-edge charge; it is diagonal in fusion-path
coordinates but need not commute with the hoop.  The next section identifies
the hoop sector of the unique ground state and computes the cut-projector
compression in that sector.

\subsection{Parity-complete cut-charge theorem}
\label{sec:main-theorem}

Abbreviate
\begin{equation}
  \delta=\cD^{-2},
  \qquad
  c=1-\delta=\frac{\varphi^2}{\cD^2}.
  \label{eq:delta-c}
\end{equation}

\begin{theorem}[Parity-complete Fibonacci cut charge]
\label{thm:parity-cut-charge}
For odd \(L\), let \(B_L\) be the cut-active projection constructed in
proposition~\ref{prop:odd-cut-algebra}.  Then for every periodic \(L\geq3\)
and every choice of fusion edge as the cut,
\begin{equation}
  \chi_L=
  \begin{cases}
    \delta,&L\ \text{even},\\
    cq_L,&L\ \text{odd},
  \end{cases}
  \qquad
  q_L:=\langle\Omega_L,B_L\Omega_L\rangle\in[0,1].
  \label{eq:parity-cut-probability}
\end{equation}
Equivalently, the two charge probabilities are
\begin{equation}
  (P_{\mathbf 1},P_\tau)=
  \begin{cases}
    (\delta,c),&L\ \text{even},\\
    (cq_L,1-cq_L),&L\ \text{odd}.
  \end{cases}
  \label{eq:two-charge-probabilities}
\end{equation}
\end{theorem}

At even length the theorem fixes the two charge probabilities.  At odd
length it reduces the exact response to the ground-state Born coordinate
\(q_L\).  Section~\ref{sec:parity-mechanism} proves both statements by first
selecting the physical hoop sector and then compressing the cut projector.
Section~\ref{sec:physical-response} realizes both branches through one
Perron-root derivative, section~\ref{sec:odd-coordinate} identifies the
information retained
by each compressed algebra, and section~\ref{sec:exact-odd-evaluation}
provides exact finite-dimensional access to the odd coordinate.

%% file: sections/parity-mechanism.tex
\section{Parity-selected sectors and cut compression}
\label{sec:parity-mechanism}

The parity dependence in theorem~\ref{thm:parity-cut-charge} follows from
sector selection and cut compression in that causal order.  The state input
is the unique strictly positive antiferromagnetic ground ray in the
fusion-path basis.  The operator input is the parity sign of a particular hoop
column, obtained by applying \(Y_{\tau,L}\) to the all-\(\tau\) word.
Simplicity and hoop commutation place the ground ray in one hoop eigenspace,
while this column identity and positivity determine which eigenspace contains
it.  The projectors \(\Pi_L\) form a length-indexed family defined by the same
cut measurement rule.  After sector selection, their compression is
scalar in the positive sector and decomposes into two character components in
the complementary sector.  We establish these statements in that order.

\subsection{Positive AFM ground ray and hoop-column parity sign}
\label{sec:parity-selection}

\begin{lemma}[Positive antiferromagnetic ground ray]
\label{lem:positive-ground}
For \(L\geq3\), \(H_L^{\mathrm{AFM}}\) has a unique ground ray, represented
by a vector that is strictly positive in the fusion-path basis.
\end{lemma}

\begin{proof}
In every nontrivial local two-state block, the off-diagonal entries of
\(-e_j\) are strictly negative.  The corresponding configuration graph is
connected.  An isolated vacuum label, necessarily flanked by two \(\tau\)
labels, can be replaced by \(\tau\).  Repeating this move reaches the
all-\(\tau\) word, and reversibility connects every admissible word.  A scalar
shift of \(-H_L^{\mathrm{AFM}}\) is therefore an irreducible nonnegative
matrix.
Perron--Frobenius gives a simple strictly positive maximal eigenvector,
which is the ground vector of \(H_L^{\mathrm{AFM}}\).
\end{proof}

The column sign can be isolated on one reference word.  Write
\begin{equation}
  \ket{T}:=\ket{\tau\cdots\tau}.
  \label{eq:all-tau-state}
\end{equation}
This basis word is admissible at every length.  It identifies the hoop column
evaluated below, whose sector projections provide the strictly positive test
vectors used in the ground-state comparison.

\begin{lemma}[Parity sign of the hoop column]
\label{lem:hoop-column-sign}
For every admissible output word \(y\),
\begin{equation}
  \operatorname{sgn}\bra{y}Y_{\tau,L}\ket{T}=(-1)^L,
  \qquad
  \bra{y}Y_{\tau,L}\ket{T}\neq0.
  \label{eq:hoop-column-sign}
\end{equation}
\end{lemma}

\begin{proof}
The only negative local factor in this column is the
\(-\varphi^{-1}\) entry of \eqref{eq:fibonacci-F}, associated with an
adjacent \(\tau\tau\) edge.  If \(z(y)\) counts the vacuum labels in \(y\),
cyclic admissibility gives
\begin{equation}
  N_{\tau\tau}(y)=L-2z(y).
  \label{eq:tau-tau-count}
\end{equation}
Fix an admissible output word \(y\).  The matrix element is a finite sum over
the intermediate labels of the hoop contraction with input \(\ket{T}\), and
those labels are precisely the components of the output word.  Fixing \(y\)
therefore leaves no free summation index, so exactly one contribution
survives, and admissibility of \(y\) makes each of its local factors nonzero.
Its sole negative local factor occurs once at each of the
\(N_{\tau\tau}(y)\) adjacent \(\tau\tau\) edges of \(y\) and nowhere else,
while every remaining local factor is positive.  That contribution, and hence
the matrix element, is therefore nonzero with sign
\((-1)^{N_{\tau\tau}(y)}=(-1)^L\).  The argument is uniform in \(L\) and in
\(y\).
\end{proof}

\subsection{Physical sector selection}
\label{sec:physical-sector-selection}

The sector argument combines simplicity with the column sign.  Simplicity
and hoop commutation place \(\Omega_L\) in one of the two orthogonal hoop
eigenspaces.  At each parity the column sign and the projector formulas
produce a strictly positive vector in one of those eigenspaces.  Its overlap
with the strictly positive ground vector is nonzero, so orthogonality excludes
the other sector assignment.

\begin{proposition}[Physical hoop sector]
\label{prop:physical-hoop-sector}
For every \(L\geq3\),
\begin{equation}
  \Omega_L\in
  \begin{cases}
    \Ran E_{+,L},&L\ \text{even},\\
    \Ran E_{-,L},&L\ \text{odd}.
  \end{cases}
  \label{eq:parity-sector-selection}
\end{equation}
\end{proposition}

\begin{proof}
Equation~\eqref{eq:hoop-algebra} and the simplicity in
lemma~\ref{lem:positive-ground} place \(\Omega_L\) in exactly one hoop
eigenspace.  If \(L\) is even, lemma~\ref{lem:hoop-column-sign} makes
\(Y_{\tau,L}\ket{T}\) componentwise positive, hence
\(E_{+,L}\ket{T}\) is strictly positive.  A positive \(\Omega_L\) cannot be
orthogonal to that vector, so it cannot lie in \(\Ran E_{-,L}\).

If \(L\) is odd, the same hoop column is componentwise negative.  The
projector formula \eqref{eq:hoop-projectors} gives
\begin{equation}
  E_{-,L}\ket{T}
  =
  \cD^{-2}
  \bigl(\varphi^2\Id-\varphi Y_{\tau,L}\bigr)\ket{T},
  \label{eq:odd-positive-test-vector}
\end{equation}
which is strictly positive.  Orthogonality now excludes
\(\Omega_L\in\Ran E_{+,L}\).
\end{proof}

Equation~\eqref{eq:hoop-column-sign} is a matrix-column identity in the real
fusion-path gauge fixed above.  Combined with strict positivity in the same
basis, it gives the sector assignment above.  Across the length-indexed AFM
family, parity determines whether
\(\Omega_L\) belongs to \(\Ran E_{+,L}\) or \(\Ran E_{-,L}\).  Each \(L\)
has its own finite Hamiltonian and cut projector on \(\widehat{\cH}_L\).  The
Hamiltonian prescription, cut measurement rule and charge alphabet remain
fixed across the family.  With the sector selected, we now determine the
compression of the cut projector.

\subsection{Scalar compression at even length}
\label{sec:even-scalar-compression}

We first analyze the positive sector independently of which parity selects
it.  Suppress the length and cut labels temporarily and write
\(\Pi_0=\Pi_L\), \(\Pi_\tau=\Id-\Pi_0\).  Relative to
\(\widehat{\cH}_L=\cH_{\mathbf 1}^{(\ell)}\oplus
\cH_\tau^{(\ell)}\), set
\begin{equation}
  A=\Id+\varphi Y_\tau=\cD^2E_+.
  \label{eq:A-definition}
\end{equation}
The local input is the vacuum corner of \(A\).  A hoop crossing a vacuum cut
changes the local channel from \(\mathbf 1\) to \(\tau\), so the hoop has no
vacuum-to-vacuum block and only the identity term remains:
\begin{equation}
  A_{00}=\Id_{\cH_{\mathbf 1}^{(\ell)}}.
  \label{eq:local-vacuum-block}
\end{equation}

A global rank identity completes this local block information.
The trace and rank calculation in
appendix~\ref{app:projection-geometry} shows that the positive sector has
the same dimension as \(\cH_{\mathbf 1}^{(\ell)}\).  Together with the
vacuum-corner identity, this fixes the graph-projector geometry used in the
next proposition.

\begin{proposition}[Scalar cut-projector compression in the positive hoop sector]
\label{prop:positive-sector-cut}
For every \(L\) and every cut position,
\begin{align}
  E_+\Pi_0E_+&=\delta E_+,
  &
  E_+\Pi_\tau E_+&=cE_+,
  \label{eq:positive-sector-sandwich}
  \\
  \Pi_0E_+\Pi_0&=\delta\Pi_0.
  \label{eq:reverse-sandwich}
\end{align}
\end{proposition}

\begin{proof}
Equations~\eqref{eq:local-vacuum-block} and the rank identity proved in
appendix~\ref{app:projection-geometry} make \(\Ran E_+\) the graph of an
isometry scaled by \(\varphi\).  The resulting orthogonal graph projector
gives \eqref{eq:positive-sector-sandwich} and \eqref{eq:reverse-sandwich}.  The
complete block calculation is recorded there.
\end{proof}

This is a sector statement valid at every length and for every cut position.
For even \(L\), proposition~\ref{prop:physical-hoop-sector} places the
physical state in \(\Ran E_+\).  Taking its expectation in
\eqref{eq:positive-sector-sandwich} gives
\begin{equation}
  P_{\mathbf 1}=\frac1{\cD^2},
  \qquad
  P_\tau=\frac{\varphi^2}{\cD^2},
  \qquad
  \frac{P_\tau}{P_{\mathbf 1}}=\varphi^2.
  \label{eq:even-fibonacci-weights}
\end{equation}
Because the compressed operator is scalar, no coordinate of a normalized
state inside the selected positive sector affects these probabilities.
After physical sector selection, the even response is therefore fixed by the
Fibonacci categorical coefficients.  The same ratio defines the exact
finite-lattice comparison quantity
\begin{equation}
  \ln g_{\mathrm{cut}}
  :=
  \frac12\ln\frac{P_\tau}{P_{\mathbf 1}}
  =
  \ln\varphi.
  \label{eq:operational-g-cut}
\end{equation}
This finite-lattice cut ratio supplies a concrete target for a continuum
comparison once a scaling limit and boundary-state dictionary are specified;
the identity established here concerns the finite periodic chain.

\subsection{Two-character compression at odd length}
\label{sec:odd-two-characters}

At odd length the selected ground state lies in the complementary sector,
where the compression of the same cut projector decomposes into two
projector components rather than being scalar.  Keep the abbreviations of the
previous subsection and set \(K=E_+\), \(E=E_-=\Id-K\), and \(\Pi=\Pi_0\).
Define
\begin{equation}
  B=c^{-1}E\Pi E,
  \qquad
  P_{\mathrm{ex}}=E-B.
  \label{eq:odd-character-projectors}
\end{equation}
The following proposition identifies \(B\) as the cut-active component and
\(P_{\mathrm{ex}}\) as its orthogonal complement inside the sector.

\begin{proposition}[Odd two-character cut algebra]
\label{prop:odd-cut-algebra}
The operators \(B\) and \(P_{\mathrm{ex}}\) are orthogonal projections.
Moreover,
\begin{equation}
  E\Pi E=cB,
  \qquad
  E=B+P_{\mathrm{ex}},
  \label{eq:odd-sector-compression}
\end{equation}
and
\begin{equation}
  \operatorname{rank}B=\operatorname{Fib}_{L-1},
  \qquad
  \operatorname{rank}P_{\mathrm{ex}}=\operatorname{Fib}_{L}.
  \label{eq:odd-character-ranks}
\end{equation}
For odd \(L\geq3\), both characters are nonzero.
\end{proposition}

\begin{proof}
Using \eqref{eq:reverse-sandwich},
\begin{equation}
  (E\Pi E)^2
  =
  E\Pi(\Id-K)\Pi E
  =
  cE\Pi E.
  \label{eq:odd-projection-identity}
\end{equation}
Thus \(B^2=B=B^\dagger\), and \(P_{\mathrm{ex}}\) is its orthogonal
complement inside \(E\).  Appendix~\ref{app:projection-geometry}
constructs the cut partial isometry and proves the two rank formulas.
\end{proof}

The two compression results now give the physical compressed cut algebra:
\begin{equation}
  \cA_{\mathrm{cut},L}^{\mathrm{phys}}
  =
  \begin{cases}
    \mathbb C E_{+,L},&L\ \text{even},\\
    \mathbb C B_L\oplus\mathbb C P_{\mathrm{ex},L},&L\ \text{odd}.
  \end{cases}
  \label{eq:physical-cut-algebra}
\end{equation}
It is one-dimensional at even length and two-dimensional at odd length.
These algebras belong to different members of the length-indexed family:
the Hamiltonian prescription, ground-state prescription, cut measurement rule and
charge alphabet remain fixed, while parity selects the sector in which the
cut is compressed.  The rank formulas confirm that both odd components are
nonzero.  For the odd ground state,
proposition~\ref{prop:physical-hoop-sector} and
\eqref{eq:odd-sector-compression} give
\begin{equation}
  \langle\Omega_L,\Pi_L\Omega_L\rangle
  =
  c\langle\Omega_L,B_L\Omega_L\rangle
  =
  cq_L.
  \label{eq:odd-ground-cut-weight}
\end{equation}
Together these identities prove theorem~\ref{thm:parity-cut-charge}.  Scalar
compression fixes the even response completely.  At odd length the exact
two-character compression reduces the response to the ground-state weight
\(q_L\) of the cut-active projector \(B_L\), and the vacuum-charge probability
is \(cq_L\).  Thus compression fixes the form and information content of the
odd law while the physical state supplies its one surviving coordinate.  The
next section realizes this already-established response through one marked
Perron-root construction.

%% file: sections/physical-response.tex
\section{A common Perron-root representation across both parities}
\label{sec:physical-response}

Theorem~\ref{thm:parity-cut-charge} establishes the exact two-branch
cut-charge law, and section~\ref{sec:parity-mechanism} derives its
sector-selection and compression mechanism.  This section supplies a common
constructive representation of that established response.  For each fixed
length, it preserves the Fibonacci chain, its prescribed antiferromagnetic
ground state, and the cut while changing only their representation.
The construction is a length-indexed family: each chain Hilbert space is
mapped to its corresponding annular \(A_4\) height space, whose seam explicitly
retains the length parity.  This length-indexed family covering both parities
represents the same physical Born response through one \(A_4\) representation,
state bridge, and Perron-root derivative at each length.

\subsection{The \texorpdfstring{$A_4$}{A4} representation for both parities}
\label{sec:parity-carrier}

The representation uses height configurations on a four-vertex graph.  Label the
vertices of the \(A_4\) graph \(1,2,3,4\), let \(\rho(h)=5-h\) be the
reflection exchanging the two ends, and collect the Perron--Frobenius
weights together with the length-dependent seam:
\begin{equation}
  \psi=(1,\varphi,\varphi,1),
  \qquad
  \Theta_L=\rho^L.
  \label{eq:A4-data}
\end{equation}
The \(A_4\) height space (the carrier space for this representation) is
\begin{equation}
  \mathsf M_L^{\mathrm{pc}}
  =
  \operatorname{span}
  \left\{
    (h_0,\ldots,h_L):
    |h_{j+1}-h_j|=1,\quad
    h_L=\Theta_L(h_0)
  \right\}.
  \label{eq:parity-complete-carrier}
\end{equation}
Heights step by one, while the closing condition is explicitly
length-dependent: even length has the identity seam and odd length has the
graph-reflection seam.  The identity- and reflection-seamed closures lie in
the standard periodic \(A_4\) Temperley--Lieb/RSOS framework
\cite{AndrewsBaxterForrester1984,Pasquier1987,
ChuiMercatPearce2003,BelleteteEtAl2023}.  The seam records which member of
the representation family is used; the quotient below retains the same
Fibonacci microscopic model at both parities.

At each fixed length, quotienting height configurations by global reflection
recovers the Fibonacci fusion-path basis.  Let \(\kappa_L\) reflect every
height at once.  The orbit quotient
\begin{equation}
  \{1,4\}\longmapsto\mathbf 1,
  \qquad
  \{2,3\}\longmapsto\tau
  \label{eq:A4-orbit-quotient}
\end{equation}
has adjacency matrix \eqref{eq:fusion-adjacency}.  Each Fibonacci word has
two reflection-related height lifts: if \(a(x)\) is either lift of a
Fibonacci word \(x\), the other is \(\rho a(x)\).

\begin{proposition}[Orbit isometry for both parities]
\label{prop:orbit-isometry}
For every \(L\geq2\),
\begin{equation}
  U_L:\widehat{\cH}_L
  \xrightarrow{\ \cong\ }
  \ker(\kappa_L-\Id)\subset\mathsf M_L^{\mathrm{pc}},
  \qquad
  U_L\ket{x}
  =
  \frac{\ket{a(x)}+\ket{\rho a(x)}}{\sqrt2},
  \label{eq:orbit-isometry}
\end{equation}
is an isometry.  It intertwines the local generators, the twisted
translation, and every cut position:
\begin{align}
  e^{A_4}_{L,j}U_L&=U_Le^{\mathrm{Fib}}_{L,j},
  &
  \mathsf s_L^{A_4}U_L&=U_L\mathsf s_L^{\mathrm{Fib}},
  \label{eq:orbit-TL-translation}
  \\
  \widetilde\Pi_{L,j}U_L&=U_L\Pi_{L,j},
  &
  \widetilde\Pi_{L,j}
  &:=
  \mathbf 1_{\{h_j\in\{1,4\}\}}.
  \label{eq:orbit-cut-intertwiner}
\end{align}
\end{proposition}

\begin{proof}
Appendix~\ref{app:carrier-covariance} constructs the two reflection-related
lifts and checks the generator, translation, and cut-intertwining identities
coefficient by coefficient.
\end{proof}

The cut intertwiner fixes the physical observable dictionary.  For any
normalized chain state, the expectation of \(\Pi_{L,j}\) agrees with the
expectation of \(\widetilde\Pi_{L,j}\) in the state transported by \(U_L\).
This follows from the isometry and cut intertwiner before specializing the
transported state.  Thus \(\widetilde\Pi_{L,j}\) is the height-space
representation of the already defined chain projector, rather than a separately
chosen height-space observable.  Appendix~\ref{app:carrier-covariance} records the
braid/hoop convention and neutral-insertion covariance, including transport
of the mark inside
\(\operatorname{span}\{\Id,\widetilde\Pi\}\).  At every length, the
construction represents the same uncompressed cut observable and its
Born dictionary.  After physical-sector selection, compressing that observable
gives the parity-dependent physical algebra
\eqref{eq:physical-cut-algebra}.

\subsection{The Perron row state is the AFM ground state}
\label{sec:perron-ground}

A common response also requires the height-space vector representing the physical
AFM ground state.  The homogeneous face rows provide this state bridge.  They
are assembled from the critical face weight
\begin{equation}
  W(u)=A(u)\Id+B(u)e,
  \qquad
  A(u)=\frac{\sin(\lambda-u)}{\sin\lambda},
  \qquad
  B(u)=\frac{\sin u}{\sin\lambda},
  \qquad
  \lambda=\frac{\pi}{5}.
  \label{eq:critical-face-weight}
\end{equation}
For fixed \(L\), let \(t_L(u)\) denote
the homogeneous face row on \(\mathsf M_L^{\mathrm{pc}}\).  The Yang--Baxter
relations give \([t_L(u),t_L(v)]=0\), so these rows form a commuting family
on that fixed height space.  Select its isotropic member:
\begin{equation}
  u_*=\frac{\lambda}{2},
  \qquad
  T_L=t_L(u_*).
  \label{eq:isotropic-row}
\end{equation}
The vector \(v_L\) below is the unit representative, chosen with positive
coordinates, of the simple Perron ray of \(T_L\).

\begin{theorem}[Perron row equals the physical ground state]
\label{thm:perron-ground-identity}
With positive normalization,
\begin{equation}
  v_L=U_L\Omega_L
  \label{eq:perron-ground-identity}
\end{equation}
for every \(L\geq3\).
\end{theorem}

\begin{proof}
Appendix~\ref{app:perron-response} proves that all rows in the physical
interval share one positive ray and that their regular-point logarithmic
derivative is the AFM Hamiltonian up to scalar and normalization.  This ray
and the AFM ground ray are both simple and strictly positive, so they
coincide under \(U_L\).
\end{proof}

Equation~\eqref{eq:perron-ground-identity} is the state bridge of the
construction.  Mathematically, it equates two normalized vectors: the
positive Perron eigenvector \(v_L\) and the \(U_L\)-transport of the AFM
ground-state vector \(\Omega_L\).  Physically, those vectors represent the
same state under the isometry.  The logarithmic-derivative relation in the
proof identifies the corresponding rays; simplicity and strict positivity
select the common ray uniquely, and the common positive unit normalization
fixes the vector equality.

\subsection{Derivative of the marked Perron root}
\label{sec:marked-response}

With the observable dictionary and state bridge established in that order,
the response is obtained by weighting the transported mark and
differentiating.  Suppress the cut position and write
\(\widetilde\Pi_L\) for its image under
\eqref{eq:orbit-cut-intertwiner}.  Define the positive congruence
\begin{equation}
  \mathsf Z_L(s)
  =
  \exp\!\left(\frac{s}{2}\widetilde\Pi_L\right),
  \qquad
  T_L(s)=\mathsf Z_L(s)T_L\mathsf Z_L(s),
  \label{eq:marked-row-deformation}
\end{equation}
and let \(\Lambda_L(s)\) be its simple Perron root for real \(s\) near
zero.  Because \(\mathsf Z_L(s)\) is positive diagonal in the height basis,
the congruence preserves the nonnegative zero pattern and irreducibility of
\(T_L\); Perron--Frobenius theory therefore supplies the simple root in the
stated neighbourhood.  The \(s\)-family is a response deformation rather
than another Yang--Baxter spectral family: its role is differentiation at
\(s=0\), and no commutativity between distinct values of \(s\) is invoked.

\begin{theorem}[Marked Perron-root derivative]
\label{thm:marked-perron-response}
For every \(L\geq3\),
\begin{equation}
  \left.
  \frac{\partial}{\partial s}\log\Lambda_L(s)
  \right|_{s=0}
  =
  v_L^{\mathsf T}\widetilde\Pi_Lv_L
  =
  \langle\Omega_L,\Pi_L\Omega_L\rangle
  =
  \chi_L.
  \label{eq:marked-perron-response}
\end{equation}
\end{theorem}

\begin{proof}
The symmetric Perron-root differentiation and the transport back through
\(U_L\) are given in appendix~\ref{app:perron-response}.
\end{proof}

The response \eqref{eq:marked-perron-response} holds on both parity towers
and realizes the complete law of
theorem~\ref{thm:parity-cut-charge} through the length-indexed family
\(U_L\).  The \(A_4\) representation, state bridge, and Perron-root derivative form one
subordinate construction: at each fixed length, the derivative returns the
Born probability of the original cut projector in the physical ground state.
The two branches thereby acquire a common representation while retaining
their parity-dependent seams and their previously established
sector-compression mechanism.  Section~\ref{sec:odd-coordinate} now
interprets the information retained by the scalar and two-character
compressed algebras and gives the odd coordinate a height-space
representation.

%% file: sections/odd-coordinate.tex
\section{Categorical rigidity and state sensitivity}
\label{sec:odd-coordinate}

Theorem~\ref{thm:parity-cut-charge} states the exact two-branch law,
section~\ref{sec:parity-mechanism} derives it from sector selection followed
by cut compression, and section~\ref{sec:physical-response} realizes it at
every length through one Perron-root derivative.  What remains is the physical
reading of that mechanism: how much a compressed cut algebra can still say
about the state it is measured in.  The answer splits along parity.  The even
algebra is scalar, and the law it returns is fixed by the category alone; the
odd algebra carries two characters, and one number belonging to the ground
state survives inside it.  This section makes both statements precise.

\subsection{Scalar compression and nonzero raw-projector variance}
\label{sec:even-response-meaning}

Fix an even length \(L\) and let \(\psi\) be any normalized vector in the
selected sector \(\Ran E_{+,L}\).  Proposition~\ref{prop:positive-sector-cut}
and idempotency \(\Pi_L^2=\Pi_L\) fix both cut-projector moments at once:
\begin{equation}
  \langle\psi,\Pi_L\psi\rangle=\delta,
  \qquad
  \operatorname{Var}_{\psi}(\Pi_L)
  =
  \delta-\delta^2
  =
  \delta c.
  \label{eq:even-cut-moments}
\end{equation}
Together the two moments separate scalar compression from raw-projector
determinism.  Since \(\Pi_L\) is a projector, its measurement has only the
two spectral outcomes, and the fixed mean therefore fixes their probabilities.
At the same time, the variance
\(\delta c\) is strictly positive, so \(\psi\) is not an eigenvector of the
raw projector \(\Pi_L\): individual measurements of this cut still
fluctuate.  Every normalized state in the positive hoop sector consequently
gives the same binary distribution for the cut, including the physical
ground state, and that distribution is the even branch of
\eqref{eq:two-charge-probabilities}.  The Fibonacci weights, and with them
the finite-lattice value \(\ln g_{\rm cut}=\ln\varphi\) of
\eqref{eq:operational-g-cut}, therefore fully determine this cut-charge law on
the even branch.  Inside \(\Ran E_{+,L}\) the compressed observable has no
further state coordinate to report; the states retain other degrees of
freedom that this particular cut does not resolve.

The construction of section~\ref{sec:marked-response} realizes this rigidity
on the physical ground state.  Its Perron derivative
\eqref{eq:marked-perron-response} is defined by the physical Perron ray and
returns the category-fixed value for that state.  Sector-wide universality,
for arbitrary normalized vectors in \(\Ran E_{+,L}\), follows instead from
the scalar compression above.  Category, parity-driven sector selection and
cut geometry therefore close the even problem without an additional
state-dependent coordinate.

\subsection{Odd character weights versus charge outcomes}
\label{sec:odd-character-weights}

At odd \(L\), proposition~\ref{prop:physical-hoop-sector} selects the
complementary sector \(\Ran E_{-,L}\), and there the ground state keeps a
coordinate of its own.  That coordinate and the two pairs of weights it
controls are
\begin{equation}
  \begin{aligned}
    q_L
    &=
    \langle\Omega_L,B_L\Omega_L\rangle
    =
    \lVert B_L\Omega_L\rVert^2,
    \\
    \bigl(
      \langle B_L\rangle_{\Omega_L},
      \langle P_{\mathrm{ex},L}\rangle_{\Omega_L}
    \bigr)
    &=(q_L,1-q_L),
    &
    (P_{\mathbf 1},P_\tau)
    &=(cq_L,1-cq_L).
  \end{aligned}
  \label{eq:odd-characters-versus-cut-outcomes}
\end{equation}
The two pairs answer different questions.  The pair \((q_L,1-q_L)\) consists
of character weights: they are the ground-state expectations of the
projectors \(B_L\) and \(P_{\mathrm{ex},L}\) of the odd cut algebra.  The
pair \((cq_L,1-cq_L)\) gives the charge outcome probabilities carried by
\(\mathbf 1\) and \(\tau\).  The category-fixed compression coefficient \(c\)
in \eqref{eq:odd-sector-compression} belongs to the cut-projector compression
rather than to the state.  It sends the \(B_L\) character weight \(q_L\) to
the \(\mathbf 1\)-charge probability \(cq_L\), while normalization fixes the
\(\tau\)-charge probability as \(1-cq_L\).  Odd length changes the cut algebra
inside the selected sector from a scalar algebra to a two-character algebra,
and that change carries the contrast: the even branch shows categorical
rigidity, whereas the odd branch is sensitive to the state through \(q_L\).

\subsection{Limits of central data}
\label{sec:central-blindness}

The next question is what central periodic-TL or fusion-endomorphism data can
determine about \(q_L\).  The selected physical hoop sector is
\begin{equation}
  \mathcal H_{22}[L]
  =
  \Ran E_{-,L}
  \cong
  \overline W_{(1/2,1)}(L).
  \label{eq:selected-odd-packet}
\end{equation}
Appendix~\ref{app:odd-packet} identifies it as a simple module, from the
object-\(3\) folding, the oriented hoop characters and the nonzero
one-through-line action
\cite{GrahamLehrer1998,Belletete2020fusion,BelleteteEtAl2023}.  Simplicity
has an immediate consequence: every central periodic-TL or fusion
endomorphism acts on this simple module as a scalar, and a scalar assigns the same
value to every state in it.

Two states make that concrete.  Choose unit vectors \(b\in\Ran B_L\) and
\(p\in\Ran P_{\mathrm{ex},L}\), available by
proposition~\ref{prop:odd-cut-algebra}, and denote the associated vector
states by \(\omega_B\) and \(\omega_{\mathrm{ex}}\).  For every central
endomorphism \(Z\),
\begin{equation}
  \omega_B(Z)=\omega_{\mathrm{ex}}(Z),
  \qquad
  \omega_B(B_L)=1,
  \qquad
  \omega_{\mathrm{ex}}(B_L)=0.
  \label{eq:central-blindness-counterexample}
\end{equation}
The exhibited pair has a precise consequence inside
\(\mathcal H_{22}[L]\): central-endomorphism data alone cannot support a
universal rule that recovers the within-sector expectation of \(B_L\) for
every state in the selected simple odd module.  State-specific input may still
determine the
physical ground-state coordinate \(q_L\).  The representation-theoretic proof
is in appendix~\ref{app:odd-packet}.  Marked and otherwise noncentral
observables supply data outside this scope, and the next subsection gives one
such height-space representation.

\subsection{Height-space representation of the odd coordinate}
\label{sec:pointed-height-placement}

That observable comes from the height picture on a distinct height space.  The
isometry \(U_L\) of section~\ref{sec:parity-carrier} belongs to the length-indexed
family covering both parities, and it already maps the physical chain, the cut
and the Perron state to the \(A_4\) height space.  What is used here is a
separately typed map.  Within the standard face/height lineage
\cite{AndrewsBaxterForrester1984,OwczarekBaxter1987,
FinchFlohrFrahm2014}, appendix~\ref{app:pointed-height} constructs the
unitary \(V_L^{\rm ht}\) from \(\widehat{\cH}_L\) onto the reflection-even
periodic object-\(3\) height space.  This height space is distinct from the
target of \(U_L\).  The restriction of \(V_L^{\rm ht}\) maps
\(\Ran E_{-,L}\) unitarily onto its image in that height space, so a mark
assembled on the image transports back to the selected odd sector
\(\mathcal H_{22}[L]\).  Let \(H_j\) and \(R\) be its cut height and
reflection, set \(D_j(t)=t^{2H_j-5}\), \(\Pi_+=(\Id+R)/2\), and
\(c_r(t)=(t^r+t^{-r})/2\).  Write
\(\Pi_{22}^{\rm ht}=V_L^{\rm ht}E_{-,L}(V_L^{\rm ht})^{-1}\) and let
\(I_{{\rm odd},L}\) be the identity on \(\Ran E_{-,L}\).

\begin{theorem}[Height-space operator family at a marked site]
\label{thm:physical-pointed-height-pencil}
For \(t>0\), the mark on the selected odd sector
\begin{equation}
  K_{L,j}(t)
  =
  (V_L^{\rm ht})^{-1}
  \Pi_{22}^{\rm ht}\Pi_+D_j(t)\Pi_+\Pi_{22}^{\rm ht}
  V_L^{\rm ht}
  \label{eq:physical-pointed-height-definition}
\end{equation}
is a Hermitian observable on the selected odd sector and satisfies
\begin{equation}
  K_{L,j}(t)
  =
  c_1(t)I_{{\rm odd},L}
  +
  c\bigl[c_3(t)-c_1(t)\bigr]B_{L,j}.
  \label{eq:physical-pointed-height-pencil}
\end{equation}
For the reference cut \(B_{L,j}=B_L\), and therefore
\begin{equation}
  \langle\Omega_L,K_{L,j}(t)\Omega_L\rangle
  =
  c_1(t)
  +
  c\bigl[c_3(t)-c_1(t)\bigr]q_L.
  \label{eq:physical-two-character-response}
\end{equation}
\end{theorem}

\begin{proof}
Appendix~\ref{app:pointed-height} proves the reflection compression and
transports the cut projector through \(V_L^{\rm ht}\).  Substitution gives
\eqref{eq:physical-pointed-height-pencil}; its ground-state expectation is
\eqref{eq:physical-two-character-response}.
\end{proof}

The theorem law \eqref{eq:two-charge-probabilities} is the physical
conclusion, and the Perron-root derivative \eqref{eq:marked-perron-response} is its
common realization across both parities.  The height-space operator family
\eqref{eq:physical-pointed-height-pencil} gives a separate representation of
the odd Born coordinate through its physical ground-state expectation.  The
observable is built from the sector identity and the odd character projector
with state-independent coefficients; \(q_L\) appears only after taking that
expectation in \(\Omega_L\).  Appendix~\ref{app:pointed-height} records the
raw height effects, the cut-projector covariance and a related ambient vertex
identity, making explicit that the ambient identity does not define the
physical dictionary.
Section~\ref{sec:exact-odd-evaluation} then evaluates \(q_L\) directly on the
physical fusion-path Hilbert space.

%% file: sections/exact-odd-evaluation.tex
\section{Exact finite-size evaluation of the odd coordinate}
\label{sec:exact-odd-evaluation}

At odd \(L\), the preceding section identifies \(q_L\) as the physical Born
coordinate controlling the weights in
\eqref{eq:odd-characters-versus-cut-outcomes}.  This section gives that same
coordinate an exact finite-dimensional evaluation directly on the physical
fusion-path Hilbert space, without separately constructing the ground-state
eigenvector.  At each fixed odd \(L\), the evaluator uses the sign-flipped
Hamiltonian, the cut projector \(\Pi_L\), and the category-fixed
coefficient \(c\); the characteristic polynomial, spectral root and adjugate
below are obtained from that Hamiltonian.

\subsection{Simple Perron root and physical spectral projector}

The sign flip reverses the energy ordering: the physical ground state of
\(H_L^{\mathrm{AFM}}\) is the top eigenstate of \(\mathcal K_L\).  In the
physical fusion-path basis, \(\mathcal K_L\) is nonnegative, irreducible and
Hermitian, so that top eigenvalue is its largest simple Perron root and its
positive ray is the physical ground-state ray.  Define
\begin{equation}
  \mathcal K_L:=-H_L^{\mathrm{AFM}}=\sum_{j=1}^{L}e_j,
  \qquad
  p_L(z):=\det(z\Id-\mathcal K_L),
  \qquad
  \varepsilon_L:=\max\{z\in\mathbb R:p_L(z)=0\},
  \label{eq:odd-spectral-data}
\end{equation}
and set
\begin{equation}
  \mathcal A_L:=\operatorname{adj}(\varepsilon_L\Id-\mathcal K_L).
  \label{eq:odd-spectral-adjugate}
\end{equation}
Thus \(\mathcal K_L\), \(\Pi_L\) and \(c\) are the specified finite-size
data.  From \(\mathcal K_L\) one forms \(p_L\), selects \(\varepsilon_L\)
and constructs \(\mathcal A_L\).  At the simple root,
\(\varepsilon_L\Id-\mathcal K_L\) has a one-dimensional nullspace and its
adjugate has rank one.  Hermiticity identifies the left and right null rays
with the same normalized physical ground-state ray, while the choice of the
largest root makes the nonzero scalar multiplying its projector positive.
Consequently its trace is nonzero, and the normalized adjugate is exactly
the physical ground-state projector,
\begin{equation}
  \frac{\mathcal A_L}{\Tr\mathcal A_L}
  =|\Omega_L\rangle\langle\Omega_L|.
  \label{eq:normalized-adjugate-projector}
\end{equation}
The cofactor matrix therefore encodes the orthogonal ground-state projector;
no separately computed ground-state eigenvector is required.

\subsection{Adjugate/cofactor trace representation}

Pairing that projector with the cut projector gives \(\chi_L\), and odd-sector
compression then converts that expectation into \(q_L\):
\begin{equation}
  q_L
  =\frac{1}{c}\,
    \frac{\Tr(\Pi_L\mathcal A_L)}{\Tr\mathcal A_L},
  \qquad
  \chi_L=cq_L.
  \label{eq:physical-odd-adjugate-evaluation}
\end{equation}

Equation~\eqref{eq:physical-odd-adjugate-evaluation} follows by tracing
\(\Pi_L\) against the normalized adjugate.  The trace ratio is the cut-projector
ground-state expectation \(\chi_L\), and the physical compression
\eqref{eq:odd-sector-compression} identifies that expectation as \(c\) times
the already-defined \(q_L\); division by \(c\) therefore recovers \(q_L\).
Appendix~\ref{app:odd-evaluator} gives the proof and exact controls.  The
numerator and denominator are finite matrix traces on the physical
fusion-path space; \(\mathcal K_L\), \(\Pi_L\), \(c\), \(\varepsilon_L\) and
\(\mathcal A_L\) retain the distinct input and derived roles specified above.
At each fixed odd length, the output is precisely the same \(L\)-indexed
coordinate carried by the odd branch.

\subsection{Calculational outlook}

Computationally, the representation trades a separate eigenvector calculation
for determination of the largest simple root and construction of its finite
adjugate.  Its matrix order is
\(\dim\widehat{\cH}_L=\operatorname{Fib}_{L-1}+\operatorname{Fib}_{L+1}\),
so the root and cofactor calculations still act on the full physical
fusion-path space.  For every fixed odd length, this gives exact algebraic
access to \(q_L\) and supplies finite-size calculational support for the odd
probability law.  Recurrence-based evaluation, asymptotic access, and a
continuum dictionary are separate questions taken up in the discussion.

%% file: sections/discussion.tex
\section{Discussion}
\label{sec:discussion}

For the periodic antiferromagnetic golden-chain ground-state family, with one
fusion edge chosen as the cut at each length, theorem~\ref{thm:parity-cut-charge}
establishes an exact parity-complete cut-charge law for every finite
\(L\geq3\).  At even length, the compressed cut algebra is one-dimensional,
so categorical data fix the response completely.  At odd length, the algebra
has two characters \eqref{eq:physical-cut-algebra}, and the response retains
the single physical ground-state coordinate \(q_L\).  Parity thus changes the
information carried by the unchanged measurement rule, from categorical
rigidity at even length to one-coordinate state sensitivity at odd length.  A
common Perron-root derivative formula realizes both branches within a shared
state--observable dictionary (theorem~\ref{thm:marked-perron-response}).  The
finite-dimensional adjugate representation
\eqref{eq:physical-odd-adjugate-evaluation} gives exact calculational access to
\(q_L\) on the physical fusion-path Hilbert space.  Each \(L\) has its own
Hilbert space, Hamiltonian, ground state and cut projector.  Across the family,
the local Hamiltonian prescription, ground-state selection rule, charge
alphabet and cut measurement rule remain fixed.

Two arrows account for the parity dependence across this length-indexed
family.  The first is sector selection.  Hoop commutation and ground-ray
simplicity place the physical ground ray in one hoop eigenspace.  In the fixed
real fusion-path gauge, its strict positivity and the parity sign of the hoop
column then determine which eigenspace is physical.  It is the positive sector
at even \(L\) and the complementary sector at odd \(L\).  The second arrow
compresses the corresponding length-\(L\) cut projector, whose
measurement rule is common to the family.  At even length the compression is
scalar \eqref{eq:positive-sector-sandwich}, so the category-fixed coefficients
determine the full binary law.  This does not make the raw projector
deterministic.  Its ground-state variance remains nonzero
\eqref{eq:even-cut-moments}, and individual outcomes still fluctuate.  At odd
length the compressed algebra has two characters
\eqref{eq:physical-cut-algebra}, with ground-state weights
\((q_L,1-q_L)\).  These weights are not charge outcomes.  The category-fixed
coefficient \(c\) sends the cut-active weight to the vacuum-charge probability
\(cq_L\), and normalization fixes the complementary probability \(1-cq_L\)
\eqref{eq:odd-characters-versus-cut-outcomes}.  Parity therefore selects the
sector before compression.  The common measurement rule carries no
state-dependent cut coordinate at even length and exactly one at odd length.

The even pair has a precise formula antecedent.  Density-matrix work on
integrable face models, including the level-three Fibonacci case, identifies
one-site entries in a quantum-dimension-one sector with local-height
probabilities and reports the size-independent critical \(A_4\) values
\cite{FrahmWesterfeld2021,WesterfeldEtAl2023}.  Foundational work computed
exact local-height probabilities for solvable SOS models
\cite{DateEtAl1987SOSLocalHeight}.  Later studies related local-height
probabilities to conformal invariance \cite{SaleurBauer1989LocalHeightCFT} and
developed multipoint probabilities for integrable RSOS models
\cite{LukyanovPugai1996}.  Warnaar unified thermodynamic Bethe-ansatz and
corner-transfer-matrix methods for ABF local-height probabilities
\cite{Warnaar1996ABFLocalHeight}.  Modern work constructs integrable lattice
discretizations of topological defect lines in minimal models
\cite{SinhaTavaresRoySaleur2026} and develops ADE graph-fusion and defect
frameworks \cite{PearceHeymannQuella2025ADE}.  The composite-height route links
integrable anyon chains to RSOS models \cite{KakashviliArdonne2012}.
Critical local-height analysis of that composite model supplies probability
expressions and conformal-character content \cite{NissinenArdonne2012}.
Periodic RSOS defects provide the compatible hoop and defect-sector language
\cite{Belletete2020fusion}.  Under the orbit fold
\eqref{eq:A4-orbit-quotient}, each component of the even pair
\eqref{eq:even-fibonacci-weights} is recovered exactly by summing over one
reflection orbit the one-site local-height probabilities in the
quantum-dimension-one sector, rather than height labels or linear
Perron--Frobenius weights.  The
antecedent uses thermodynamic ground-state input, and its finite-sector
factorization rests on observed structure.  Here the same pair is the even
branch of a theorem for the periodic antiferromagnetic ground-state family at
every finite \(L\geq3\).  The theorem also supplies the odd branch and one
sector-selection and compression mechanism for both parities.  Thus the
shared pair links exact settings with different states, geometries and
warrants.

The projector also has a direct antecedent.  In an open, pinned Fibonacci
chain, Chandran, Schulz and Burnell used the vacuum-label fusion-path bond
projector to study locality and eigenstate thermalization.  They derived exact
finite-size probabilities at infinite temperature and tested mid-spectrum
eigenstates numerically \cite{ChandranSchulzBurnell2016}.  In the common
fusion-path coordinate, its projector onto the trivial bond label uses the
same diagonal rule as \(\Pi_L\) in
\eqref{eq:named-cut-decomposition}.  The
corresponding Born law differs because the ensemble, boundary condition and
physical question differ from those of the periodic antiferromagnetic ground
state.  Locality also depends on the coordinates used.  The projector reads
one fusion-path bond label.  In physical anyon variables, however, that label
records the accumulated fusion charge of all anyons on one side.  That work
already describes the projector as nonlocal or quasi-local in physical
variables, rather than as a few-anyon observable.  We call it the cut
projector here to keep this distinction explicit.  The
state--measurement pair in its geometry is the unit of comparison.  Its Born law is
\eqref{eq:charge-probabilities}.

Probability enters the neighboring symmetry literature through a different
pipeline.  Local actions of non-invertible symmetries can be represented by
completely positive quantum operations, including lattice implementations of
Kramers--Wannier duality that retain twisted and untwisted charge sectors
\cite{OkadaTachikawa2024,KhanKhanMohd2024}.  Fusion-category symmetries have
likewise been formulated as isometries among twisted-sector Hilbert spaces and
as trace-preserving channels in a categorical probability-preservation
construction \cite{BartschGaiSchaferNameki2026}.  The resulting probabilities
characterize an action or channel.  Complementary structural work treats
generalized charges and symmetry-TFT data
\cite{BhardwajSchaferNameki2025ChargesII}, matrix-product-operator intertwiners
\cite{LootensEtAl2023}, operator-algebraic fusion-category symmetry on the
lattice \cite{EvansJones2025FusionCategoryLattice}, and a generalized Wigner
theorem for non-invertible symmetries
\cite{OrtizEtAl2026GeneralizedWigner}.  Measurement-and-feedback protocols
occupy a distinct operational branch in which measurement helps prepare the
state rather than read out a fixed one
\cite{LiMitraLu2026MeasurementFeedback}.  Here the preparation rule and the
cut measurement rule remain fixed across the length-indexed family.
Each \(L\) nevertheless has its own finite-chain state and Hilbert space.  The
output is the Born distribution of that readout on the corresponding ground
state.  No symmetry operation is applied in this protocol.  The marked
deformation of the Perron row instead provides a calculational response
construction for the
fixed dictionary.  It is not applied to the ground state as a physical
channel.  Its logarithmic derivative returns the original expectation value
(theorem~\ref{thm:marked-perron-response}).  Process probabilities and
fixed-readout Born probabilities therefore share an operational taxonomy
while remaining distinct objects.

Information granularity supplies another comparison.  Tube algebras and
generalized charges organize global or twisted-sector data through central
idempotents and their representations
\cite{BartschBullimoreGrigoletto2023,ShaoSorceSrivastava2025,
FerragattaVanRees2026,ChoiRayhaunZheng2024}.  Under complete-positivity and
covariance constraints, primitive central idempotents define extremal
tube-sector POVM effects and optimal one-shot sector probabilities
\cite{He2026TubeDistinguishability}.  The selected odd sector of the golden
chain provides a
separate finite-ground-state example of the global-versus-state-resolved distinction.
On that selected simple odd module, every central periodic-TL or fusion
endomorphism acts as a scalar.  The two states in
\eqref{eq:central-blindness-counterexample} therefore agree on all such
central data but give the cut-active character projector the values \(1\) and
\(0\).  No state-independent function of the central characters can recover
the expectation of \(B_L\) for every state in that selected simple odd module.  The physical
coordinate \(q_L\) is the ground-state expectation of this noncentral
projector.  Binary charge labels record measurement outcomes, and
\((q_L,1-q_L)\) are ground-state weights on the two compressed characters.
Tube-sector POVM elements are effects in a different operational algebra.

Entanglement changes the geometry as well as the information retained.  For a
spatial interval in an anyonic system, the reduced-state construction is
graded by boundary-component charges and depends on geometry, fusion-tree
embedding and superselection structure
\cite{Pfeifer2014anyonicEntanglement,KatoFurrerMurao2014,
BondersonKnappPatel2017,CornfeldEtAl2019,YaukHacklHahn2026,
YeYouXu2026BipartiteAnyonicEntanglement}.  In this periodic chain, the internal relation
\eqref{eq:charge-probabilities} establishes that the cut-charge law retains
one endpoint charge and sums over the other.  It is therefore a one-endpoint
marginal of the joint charge distribution, not the full two-endpoint reduced
state or its sector decomposition.  Categorical symmetry-resolved
entanglement instead decomposes a full reduced state using symmetric
entangling-boundary states
\cite{SauraBastida2024catSREE,Das2024twistedSREE,
ChoiRayhaunZheng2024,PearceHeymannQuella2025ADE}.  Heymann and Quella
correct the bulk projector formulas by incorporating explicit boundary
conditions at the entangling surfaces.  They test asymptotic predictions with
open critical anyonic chains used as local-Hamiltonian surrogates
\cite{HeymannQuella2025}.  Those results concern continuum or asymptotic
entanglement data with two entangling surfaces.  The present quantity is an
exact finite-\(L\) Born probability on a closed ring.  Fibonacci projectors and
quantum dimensions make the comparison useful.  A one-endpoint marginal and
its binary outcome law do not determine the full reduced density matrix or its
symmetry-resolved entropies.

The even branch defines a concrete continuum comparison target.
Equation~\eqref{eq:operational-g-cut} makes \(g_{\mathrm{cut}}\) the square
root of the exact finite-lattice cut-charge probability ratio and gives
\(g_{\mathrm{cut}}=\varphi\).  Within the present theorem it is an operational
diagnostic of the category-fixed branch.  Affleck--Ludwig boundary entropy, by
contrast, is defined from a continuum boundary factor
\cite{AffleckLudwig1991,Cardy1989}.  Defect \(g\)-functions depend on continuum defect
or boundary fusion data, and non-topological integrable defects can exhibit
oscillatory finite-size behavior
\cite{HeJiangLiu2026,GutperleEtAl2024OrbifoldDefectEntropy,
ChoiRayhaunZheng2024}.  Lattice Ising
entanglement with a topological defect also shows that zero-mode finite-size
corrections can obstruct a direct inference from continuum modular data
\cite{RoySaleur2022}.  Numerical agreement alone would therefore not identify
these quantities.  A controlled comparison must specify a
lattice-to-continuum endpoint operator, a boundary or defect state, and a
scaling geometry before taking a limit
\cite{SinhaTavaresRoySaleur2026,PearceHeymannQuella2025ADE,ZiniWang2018}.
The exact lattice value is the input
to that test, not a boundary or defect \(g\)-factor already established.

The demand for an explicit dictionary also governs any transfer of the
mechanism beyond the Fibonacci chain.  A second model begins with a physical
state or sector selector analogous to \eqref{eq:parity-sector-selection}.
Transfer then asks whether compressing one fixed cut observable on that
sector changes its character structure in a controlled way, as in
\eqref{eq:physical-cut-algebra}.  A compatible response construction must
preserve the state--observable dictionary as its control parameter varies and
recover the original expectation
(theorem~\ref{thm:marked-perron-response}).  The diagnostic is that structure
together with the number of surviving state coordinates.  Exact Ising lattice
defects establish a general lattice-defect paradigm
\cite{AasenMongFendley2016}.  For analogous selectors, periodic RSOS defects
provide hoop and sector frameworks \cite{Belletete2020fusion}, supplemented by
the affine Temperley--Lieb module framework \cite{BelleteteEtAl2023}.  The
integrable \(A_4\) density-matrix literature provides the local-height
comparator used here \cite{FrahmWesterfeld2021,WesterfeldEtAl2023}.  Related
transfer questions may be posed in lattice settings for topological sectors
and non-invertible symmetries
\cite{LootensEtAl2024,BhardwajEtAl2026LatticeModels,
EvansJones2025FusionCategoryLattice}, and in tensor-network treatments of
topological defects \cite{AasenFendleyMong2020,HauruEtAl2016}.
Fusion-category and anyon-chain constructions provide neighboring algebraic
and scaling settings
\cite{FerragattaVanRees2026,AntunesRong2026CoupledAnyonChains,ZiniWang2018}.
Several
ingredients remain specific to Fibonacci.  Its selector uses
\(\tau\otimes\tau=\mathbf1\oplus\tau\), the two hoop eigenvalues and the
\((-1)^L\) transport sign, and the construction depends further on the \(A_4\)
reflection seam and fold and on the constants \(\delta\) and \(c\).  So do the
selected odd sector and the height-space construction at the marked site.  Transfer to
a second chain therefore requires the selector, compression change and
response dictionary under explicit hypotheses.

In the Fibonacci realization, the theorem is exact for the periodic
antiferromagnetic ground-state family at every finite \(L\geq3\).  At each
length, one fusion edge is chosen as the cut, and
\eqref{eq:charge-probabilities} identifies its one-endpoint marginal as the
measured quantity.  Theorem~\ref{thm:parity-cut-charge} gives the complete
two-outcome Born response at both parities.  Only the odd branch retains a
state-dependent coordinate.  Its compressed-character weights
\((q_L,1-q_L)\) become the physical charge probabilities
\((cq_L,1-cq_L)\) after category-fixed compression and normalization.
Equation~\eqref{eq:physical-odd-adjugate-evaluation} evaluates \(q_L\) on
the physical fusion-path space from the largest simple Perron root and a
ratio of adjugate traces.  This avoids a separate ground-state eigenvector
calculation, though the evaluation still acts on the full physical space,
whose matrix order
\(\dim\widehat{\cH}_L=\operatorname{Fib}_{L-1}+\operatorname{Fib}_{L+1}\)
grows exponentially with \(L\).  The theorem and endpoint dictionary fix the
law and the measured object.  The representation gives exact
finite-dimensional access to the remaining coordinate.  The full two-endpoint
distribution, controlled continuum limits, stability under Hamiltonian
deformation and lower-cost evaluation remain four distinct targets.

The complete finite-size theorem opens several further directions.  For the
odd coordinate, the exact evaluator in
\eqref{eq:physical-odd-adjugate-evaluation} acts on an exponentially large
physical space.  The open problem is to find a more efficient exact
representation that reproduces the same finite-\(L\) values with a proved
complexity improvement over the full-space evaluation.
That same direction includes the large-\(L\) behaviour of \(q_L\), which remains
open.
The continuum meaning
of the fusion-edge marginal is a second such direction.  A controlled
lattice-to-continuum interpretation must specify a lattice endpoint operator,
its associated boundary or defect state, and the scaling geometry.  Once that
dictionary is fixed, one can ask which continuum quantity, if any, corresponds
to the exact finite-lattice ratio \(g_{\mathrm{cut}}\) in
\eqref{eq:operational-g-cut}.  The joint endpoint distribution remains a
separate, more stringent problem.  Beyond Fibonacci, the central question is
how much of the mechanism survives.  A second fusion-category or ADE chain
would need a physical-sector selector analogous to
\eqref{eq:parity-sector-selection} and a compressed physical cut algebra for
a fixed cut observable as in \eqref{eq:physical-cut-algebra}.  Its
calculational response deformation would have to preserve the
state--observable dictionary and recover the original expectation at zero
mark, as in theorem~\ref{thm:marked-perron-response}.  A realization of this
kind would distinguish categorical ingredients from features specific to
Fibonacci.

%% file: appendices/projection-geometry.tex
\section{Projection geometry, traces, and rank formulas}
\label{app:projection-geometry}

This appendix records the finite-dimensional bookkeeping used by the two
sector-compression propositions in section~\ref{sec:parity-mechanism}.  It
does not introduce another observable or another ground-state sector.

\subsection{Positive-sector rank and graph projector}

On the diagonal of the hoop, only the all-\(\tau\) path contributes, giving
\begin{equation}
  \Tr Y_{\tau,L}=(-\varphi^{-1})^L.
  \label{eq:cyclic-hoop-trace}
\end{equation}
Since \(E_+\) is an orthogonal projector,
\begin{equation}
  \operatorname{rank}E_+
  =
  \frac{\Tr N^L+\varphi(-\varphi^{-1})^L}{\cD^2}
  =
  (N^L)_{00}
  =
  \dim\cH_{\mathbf 1}^{(\ell)}.
  \label{eq:positive-sector-rank}
\end{equation}
The middle equality follows by resolving \(N^L\) into the eigenspaces
\(\varphi\) and \(-\varphi^{-1}\).

Let \(S=\Ran E_+\).  Equations~\eqref{eq:local-vacuum-block} and
\eqref{eq:positive-sector-rank} imply that the coordinate projection
\(S\to\cH_{\mathbf 1}^{(\ell)}\) is bijective.  Thus \(S\) is the graph of a
map \(T:\cH_{\mathbf 1}^{(\ell)}\to\cH_\tau^{(\ell)}\).  The upper-left
block of its orthogonal graph projector is
\((\Id+T^\dagger T)^{-1}\).  Because
\((E_+)_{00}=\delta\Id\),
\begin{equation}
  (\Id+T^\dagger T)^{-1}=\delta\Id,
  \qquad
  T^\dagger T=\varphi^2\Id.
  \label{eq:graph-isometry}
\end{equation}
Hence \(T=\varphi V\) for an isometry \(V\), and the graph-projector formula
becomes
\begin{equation}
  \cD^2E_+
  =
  \begin{pmatrix}
    \Id&\varphi V^\dagger\\
    \varphi V&\varphi^2VV^\dagger
  \end{pmatrix}.
  \label{eq:positive-graph-projector}
\end{equation}
Multiplying the blocks in \eqref{eq:positive-graph-projector} gives
\eqref{eq:positive-sector-sandwich}; its upper-left block gives
\eqref{eq:reverse-sandwich}.

\subsection{Odd-sector ranks and the full cut projector}

Use the notation \(K=E_+\), \(E=E_-\), and \(\Pi=\Pi_0\) from
section~\ref{sec:odd-two-characters}.  The partial isometry
\begin{equation}
  \mathcal V_L=(c\delta)^{-1/2}K\Pi E
  \label{eq:cut-partial-isometry}
\end{equation}
obeys
\begin{equation}
  \mathcal V_L^\dagger\mathcal V_L=B,
  \qquad
  \mathcal V_L\mathcal V_L^\dagger=K.
  \label{eq:cut-partial-isometry-support}
\end{equation}
Hence \(\operatorname{rank}B=\operatorname{rank}K=(N^L)_{00}
=\operatorname{Fib}_{L-1}\).  Since
\(\Tr N^L=\operatorname{Fib}_{L-1}+\operatorname{Fib}_{L+1}\),
the complementary rank is
\(\operatorname{rank}P_{\mathrm{ex}}=\operatorname{Fib}_L\).

The same partial isometry makes the constant-angle geometry explicit:
\begin{equation}
  \Pi=
  \begin{pmatrix}
    cB&\sqrt{c\delta}\,\mathcal V_L^\dagger\\
    \sqrt{c\delta}\,\mathcal V_L&\delta K
  \end{pmatrix}
  \quad\text{on}\quad
  E\widehat{\cH}_L\oplus K\widehat{\cH}_L.
  \label{eq:constant-angle-cut}
\end{equation}
This matrix is the projection geometry behind the scalar/two-character
compression; the physical algebra itself remains
\eqref{eq:physical-cut-algebra}.

%% file: appendices/carrier-covariance.tex
\section{\texorpdfstring{$A_4$}{A4} representation for both parities and neutral-insertion covariance}
\label{app:carrier-covariance}

This appendix supplies the representation identities used in
section~\ref{sec:parity-carrier}.  It first proves the orbit isometry and fixes
the annular braid and hoop convention.  It then establishes covariance of the
row and the uncompressed cut mark under neutral two-site insertions along each
parity tower.  Throughout, the \(A_4\) representation describes the same physical
chain and cut measurement rule.

\subsection{Orbit isometry and annular convention}

For proposition~\ref{prop:orbit-isometry}, fixing one initial height gives a
unique lift of each Fibonacci word, and global reflection gives the second.
The endpoint of either lift is \(\rho^L(h_0)\), so both parities are covered
by \eqref{eq:parity-complete-carrier}.  Distinct reflection orbits are
orthogonal, proving isometry.  The weights in \eqref{eq:A4-data} are constant
on the two orbits in \eqref{eq:A4-orbit-quotient}; substituting them in the
RSOS generator therefore reproduces the Fibonacci generator coefficient by
coefficient, in the standard critical \(A\)-type face representation
\cite{AndrewsBaxterForrester1984,Pasquier1987,OwczarekBaxter1987}.
Shifting a lifted path proves the translation identity, and
membership of \(h_j\) in the outer orbit proves the cut identity.

In the pinned annular convention the braid-transfer centre
\(F_L^{\mathrm{br}}\) is related to the physical hoop by
\begin{equation}
  F_L^{\mathrm{br}}U_L
  =
  (-1)^L U_LY_{\tau,L}.
  \label{eq:braid-hoop-convention}
\end{equation}
The sign is one local orientation factor per face; all intermediate vertex
gauges telescope around the identity or reflection closure.  The oriented
annular hoops and their affine-TL central characters use the periodic-TL
defect convention of \cite{MartinSaleur1994,BelleteteEtAl2023}.

\subsection{Stability along the two parity towers}
\label{sec:length-stability}

Write \(t_L(u_1,\ldots,u_L)\) for the inhomogeneous row built from
\eqref{eq:critical-face-weight}, with the seam fixed by
\eqref{eq:parity-complete-carrier}.  Let
\(C_{i,L}:\mathsf M_L^{\mathrm{pc}}\to
\mathsf M_{L+2}^{\mathrm{pc}}\) be the normalized Perron--Frobenius cup
that inserts a neutral segment \((d,x,d)\) at position \(i\).  Because
\(\Theta_{L+2}=\Theta_L\), this insertion stays inside the same parity
tower.  The Temperley--Lieb fusion relation
\cite{TemperleyLieb1971,PasquierSaleur1990} fuses a spectral pair
\(u,u-\lambda\) through the cup and gives the row covariance
\begin{equation}
  t_{L+2}(\ldots,u,u-\lambda,\ldots)C_{i,L}
  =
  \mu(u)C_{i,L}t_L(\ldots),
  \qquad
  \mu(u)=
  \frac{\sin u\,\sin(2\lambda-u)}{\sin^2\lambda}.
  \label{eq:row-length-covariance}
\end{equation}
This is the lattice expression of a neutral two-site insertion: it changes
the length but not the annular sector carried by the row.

The cut mark is stable under the same move.  If
\(\widetilde\Pi_i\) marks the old anchor,
\(\widetilde\Pi_x\) the inserted height, and
\(\widetilde\Pi_d\) the repeated height, then
\begin{equation}
  C_{i,L}^{\dagger}\widetilde\Pi_xC_{i,L}
  =
  \varphi^{-2}(\Id-\widetilde\Pi_i),
  \qquad
  C_{i,L}^{\dagger}\widetilde\Pi_dC_{i,L}
  =
  \widetilde\Pi_i.
  \label{eq:marked-cup-closure}
\end{equation}
Hence every neutral insertion sends the marked output back into
the marked span in the \(A_4\) representation,
\(\operatorname{span}\{\Id,\widetilde\Pi_i\}\), which is transported along
both \(L\mapsto L+2\) towers.  This covariance fixes the uncompressed mark
under length change.  At odd length, the homogeneous ground state supplies
the Born coordinate \(q_L\) between the two characters of the physical
compressed algebra \eqref{eq:physical-cut-algebra}.

%% file: appendices/perron-response.tex
\section{Perron-row and marked-root technical lemmas}
\label{app:perron-response}

This appendix supplies the two technical ingredients used in
sections~\ref{sec:perron-ground} and \ref{sec:marked-response}.  The first
subsection identifies the common positive Perron ray representing the
transported antiferromagnetic ground state.  The second derives the Hamiltonian
and marked-root identities that equate the Perron derivative with the
cut-charge Born response.

\subsection{The common positive row ray}

\begin{lemma}[Common positive row ray]
\label{lem:common-positive-row}
For \(0<u<\lambda\), the rows \(t_L(u)\) share a unique strictly
positive ray.  We denote its unit vector by \(v_L\).
\end{lemma}

\begin{proof}
The commuting-row and crossing relations are those of the critical
\(A_4\) face model \cite{AndrewsBaxterForrester1984,OwczarekBaxter1987,
Baxter1982}.  Crossing symmetry makes \(T_L\) real symmetric.  The endpoint
Perron--Frobenius factors telescope through the odd seam because
\(\psi_{\rho(h)}=\psi_h\).  Throughout \(0<u<\lambda\), every nonzero
local face weight is positive and the row support contains the twisted
shift together with all shifted local Temperley--Lieb moves.  These moves
connect every height orbit to the all-inner-orbit path and are reversible,
so every row in the open interval is irreducible and nonnegative.
Perron--Frobenius gives a unique positive ray for each row.  Commutativity
then forces the rows to share that ray.  Global reflection commutes with
the rows and maps a positive Perron vector to another one; uniqueness gives
\(\kappa_Lv_L=v_L\), hence \(v_L\in\Ran U_L\).
\end{proof}

\subsection{Hamiltonian and marked-root derivatives}

At the regular point, \(t_L(0)\) is the invertible twisted shift.  The
standard logarithmic-derivative construction of the RSOS Hamiltonian
\cite{Pasquier1987,PasquierSaleur1990} applied to
\eqref{eq:critical-face-weight} gives
\begin{equation}
  t_L'(0)t_L(0)^{-1}
  =
  -L\cot\lambda\,\Id
  +\frac1{\sin\lambda}\sum_{j=1}^{L}e_j
  =
  -L\cot\lambda\,\Id
  -\frac1{\sin\lambda}H_L^{\mathrm{AFM}}.
  \label{eq:row-Hamiltonian-derivative}
\end{equation}
The twisted shift removes the row translation before the local derivatives
are summed, so the last term includes the automorphism seam as well as the
ordinary bonds.  Differentiating row commutativity at \(u=0\) yields
\([T_L,H_L^{\mathrm{AFM}}]=0\) under the intertwiner
\eqref{eq:orbit-TL-translation}.  The AFM ground line is simple and
strictly positive by lemma~\ref{lem:positive-ground}; the isotropic row has
only the positive ray of lemma~\ref{lem:common-positive-row}.  The two
lines therefore coincide, proving theorem~\ref{thm:perron-ground-identity}.

For the deformation \eqref{eq:marked-row-deformation},
\begin{equation}
  \dot T_L(0)
  =
  \frac12
  \bigl(\widetilde\Pi_LT_L+T_L\widetilde\Pi_L\bigr).
  \label{eq:marked-row-derivative}
\end{equation}
Since \(T_L\) is symmetric and its Perron eigenvalue is simple,
Hellmann--Feynman differentiation yields
\(\dot\Lambda_L(0)=
\Lambda_L(0)v_L^{\mathsf T}\widetilde\Pi_Lv_L\).
Division by \(\Lambda_L(0)\), followed by
\eqref{eq:orbit-cut-intertwiner} and
\eqref{eq:perron-ground-identity}, proves
theorem~\ref{thm:marked-perron-response}.

%% file: appendices/odd-packet.tex
\section{Selected odd sector and limits of central data}
\label{app:odd-packet}

We spell out the premise behind \eqref{eq:selected-odd-packet}.  The only
external representation-theoretic input is the restricted \(A_4\) sector
classification.  In this classification, the oriented-hoop sectors
\(\mathcal H_{xy}[L]\) occurring in the RSOS representation are simple.  Their
fundamental hoop characters are \(2\cos(x\pi/5)\) and
\(2\cos(y\pi/5)\).  The reflection-folded odd representation contains
\begin{equation}
  \widehat{\mathcal H}_L
  \cong
  \mathcal H_{13}[L]\oplus\mathcal H_{22}[L].
  \label{eq:odd-folded-sector-decomposition}
\end{equation}
This restricted statement, rather than semisimplicity of every affine-TL
standard module at a fifth root, is what we use
\cite{MartinSaleur1994,GrahamLehrer1998,BelleteteEtAl2023,
PinetSaintAubin2023}.

Let \(Y_2\) denote the source fundamental hoop.  Under the object-\(3\)
folding of appendix~\ref{app:pointed-height}, the physical Fibonacci hoop is
the reflection-invariant crossed defect
\begin{equation}
  Y_\tau=Y_2^{(3)}=Y_2^2-\Id.
  \label{eq:physical-crossed-hoop}
\end{equation}
Its two characters in \eqref{eq:odd-folded-sector-decomposition} are
\begin{equation}
  Y_\tau|_{\mathcal H_{13}}=\varphi\Id,
  \qquad
  Y_\tau|_{\mathcal H_{22}}=-\varphi^{-1}\Id.
  \label{eq:odd-folded-hoop-characters}
\end{equation}
Comparison with \eqref{eq:hoop-projectors} identifies
\(\mathcal H_{13}=\Ran E_+\) and
\(\mathcal H_{22}=\Ran E_-\).  In particular, the odd ground ray selected in
proposition~\ref{prop:physical-hoop-sector} belongs to
\(\mathcal H_{22}\).

It remains to determine the affine-TL through-line and twist labels without
a dimension guess.  Define the one-through diagram
\begin{equation}
  m_L=e_2e_4\cdots e_{L-1}.
  \label{eq:odd-one-through-diagram}
\end{equation}
Every factor is nonnegative in the symmetric Fibonacci gauge, and the
all-\(\tau\) diagonal contribution is strictly positive.  Since
\(\Omega_L\) is strictly positive,
\(m_L\Omega_L\neq0\).  A standard module with more than one through line is
annihilated by every one-through diagram, so odd parity forces one through
line, or \(k=1/2\).

For the twist, put \(\mathfrak q=e^{i\pi/5}\) and
\(\alpha=(-\mathfrak q)^{1/2}=e^{3\pi i/5}\).  The two object-\(3\) oriented
hoop characters on \(\mathcal H_{22}\) are both \(-\varphi^{-1}\), whereas
the \((k,z)=(1/2,z)\) standard-module characters are
\begin{equation}
  Y=z\alpha+z^{-1}\alpha^{-1},
  \qquad
  \overline Y=z\alpha^{-1}+z^{-1}\alpha.
  \label{eq:odd-oriented-standard-characters}
\end{equation}
Their equality gives \(z^2=1\).  The sign \(z=1\) yields
\(2\cos(3\pi/5)=-\varphi^{-1}\), while \(z=-1\) gives the opposite sign.
Thus \(z=1\), and the affine-TL classification proves
\begin{equation}
  \mathcal H_{22}[L]
  \cong
  \overline W_{(1/2,1)}(L).
  \label{eq:odd-simple-packet-proof}
\end{equation}
This is the simple odd module selected by the odd positive ground ray.  No
second physical sector or cut has been introduced.

Simplicity and Schur's lemma give the relevant limitation of central data.
Every central periodic-TL or fusion endomorphism \(Z\) acts as a scalar on
\(\mathcal H_{22}[L]\).  Choose unit vectors
\(b\in\Ran B_L\) and \(p\in\Ran P_{\mathrm{ex},L}\), which exist by
proposition~\ref{prop:odd-cut-algebra}.  Their vector states therefore agree
on every such \(Z\), while orthogonality of the two projective ranges gives
\(\omega_B(B_L)=1\) and \(\omega_{\mathrm{ex}}(B_L)=0\).  This proves
\eqref{eq:central-blindness-counterexample}.

Consequently no state-independent function of central characters can recover
\(\omega(B_L)\) for every state in the selected simple odd module.  This
limitation concerns only central data.  The marked noncentral observable in
theorem~\ref{thm:physical-pointed-height-pencil} retains precisely the
coordinate that the centre cannot see.

%% file: appendices/pointed-height.tex
\section{Height-space maps at the marked site, covariance, and scope of the ambient identity}
\label{app:pointed-height}

This appendix assembles the typed height-space construction behind
theorem~\ref{thm:physical-pointed-height-pencil}.  It first relates the two
path spaces, then proves reflection compression for the height-space operator
family and cut-projector covariance.  The final subsection records an ambient
vertex identity on a separate representation space.  That identity delimits the
construction's type boundary and does not define the physical
state--observable dictionary.

\subsection{The two path maps}

Let \(\rho(h)=5-h\).  Write \(\mathsf H_L^{(2,\rho)}\) for the \(A_4\)
paths generated by fusion with object \(2\) and closed by \(\rho\), and
\(\mathsf H_L^{(3,{\rm per})}\) for the paths generated by object \(3\)
with periodic closure.  These are the standard restricted \(A\)-type path
spaces \cite{AndrewsBaxterForrester1984,Pasquier1987,PasquierSaleur1990}.
On the former define the basis permutation
\begin{equation}
  U_L^{\rm stag}:
  \mathsf H_L^{(2,\rho)}
  \longrightarrow
  \mathsf H_L^{(3,{\rm per})},
  \qquad
  \bigl(U_L^{\rm stag}h\bigr)_j
  =
  \rho^{\,j-1}(h_j).
  \label{eq:odd-staggered-path-map}
\end{equation}
It turns the reflection seam into periodic closure and restricts to a
unitary between the reflection-even subspaces.  On periodic object-\(3\)
paths set
\begin{equation}
  \pi(1)=\pi(4)=\mathbf 1,
  \qquad
  \pi(2)=\pi(3)=\tau,
  \qquad
  F_L^{\rm fold}
  \frac{\ket g+\ket{\rho g}}{\sqrt2}
  =
  \ket{\pi(g)}.
  \label{eq:object-three-fold}
\end{equation}
Every periodic Fibonacci path has exactly two object-\(3\) lifts related by
\(\rho\).  Consequently \(F_L^{\rm fold}\) is unitary from the
reflection-even periodic object-\(3\) path space to
\(\widehat{\mathcal H}_L\), and
\begin{equation}
  V_L^{\rm ht}:=(F_L^{\rm fold})^{-1}
  \label{eq:physical-height-unitary}
\end{equation}
is a physical Fibonacci-to-height unitary.  Both path maps preserve every
paired-height cut and intertwine the complete face row in their respective
path spaces.  In particular, \(V_L^{\rm ht}\) maps
\(\Ran E_{-,L}\) onto \(\mathcal H_{22}[L]\).

For the path identities, write \(g_j=\rho^{j-1}(h_j)\).  The relation
\(3=\rho(2)\) converts each object-\(2\) step into an object-\(3\) step.
Oddness gives \(g_{L+1}=g_1\), and
\(\rho\{1,4\}=\{1,4\}\) preserves the cut.  For the fold, the object-\(3\)
fusion rules reduce under \(\pi\) to the Fibonacci fusion rules.  Choosing
one initial lift fixes the rest, while its global reflection supplies the
second choice.
The two lifts are orthogonal and reflection has no fixed height, proving
unitarity.  Expanding the face row in local Temperley--Lieb monomials then
proves the row intertwining coefficient by coefficient.

\subsection{Reflection compression and the height-space operator family}

Let \(H_j\) be the diagonal height at the cut on the periodic
object-\(3\) height space, and define
\begin{equation}
  D_j(t)=t^{2H_j-5},
  \qquad
  P_{3,j}=\mathbf 1_{\{H_j\in\{1,4\}\}},
  \qquad
  P_{1,j}=\Id-P_{3,j},
  \qquad
  c_r(t)=\frac{t^r+t^{-r}}2.
  \label{eq:pointed-height-mark}
\end{equation}
Throughout the main text and this appendix, the physical observable
statement for this operator family is restricted to \(t>0\).  In this domain
\(D_j(t)\) is a positive Hermitian diagonal operator, and its compression and
pullback are Hermitian observables.  For complex \(t\neq0\), the same equations
define only an algebraic generating family.  They carry no Hermitian-observable
claim.
The subscripts \(3\) and \(1\) label the two reflection-paired height
orbits.  With \(\Pi_+=(\Id+R)/2\), reflection compression gives
\begin{equation}
  \Pi_+D_j(t)\Pi_+
  =
  c_3(t)P_{3,j}+c_1(t)P_{1,j}.
  \label{eq:height-reflection-compression}
\end{equation}
Put
\begin{equation}
  \Pi_{22}^{\rm ht}
  =
  V_L^{\rm ht}E_{-,L}(V_L^{\rm ht})^{-1}
  \label{eq:selected-height-projector}
\end{equation}
Transport of the cut projector by \(V_L^{\rm ht}\), together with
\eqref{eq:odd-sector-compression}, gives
\begin{equation}
  \Pi_{22}^{\rm ht}P_{3,j}\Pi_{22}^{\rm ht}
  =
  cV_L^{\rm ht}B_{L,j}(V_L^{\rm ht})^{-1}.
  \label{eq:height-cut-compression}
\end{equation}
Substituting this identity and \(P_{1,j}=\Id-P_{3,j}\) into
\eqref{eq:height-reflection-compression} proves
\eqref{eq:physical-pointed-height-pencil}.  Taking the ground-state expectation
then proves \eqref{eq:physical-two-character-response}.

The raw height-orbit effects pulled back to the selected odd sector are
\begin{equation}
  \mathsf E_{3,L,j}=cB_{L,j},
  \qquad
  \mathsf E_{1,L,j}
  =
  I_{{\rm odd},L}-cB_{L,j}
  =
  P_{\mathrm{ex},L,j}+\delta B_{L,j},
  \label{eq:raw-height-effects}
\end{equation}
whereas the two projective characters are
\((B_{L,j},P_{\mathrm{ex},L,j})\).  Thus the raw \(P_{1,j}\) orbit is not
the complementary projector \(P_{\mathrm{ex},L,j}\).

\subsection{Cut-projector covariance}

If \(S_L\) is physical one-site translation and every cut-indexed object is
defined by conjugating the reference cut, then
\begin{align}
  \Pi_{L,j+1}&=S_L\Pi_{L,j}S_L^{-1},
  &
  B_{L,j+1}&=S_LB_{L,j}S_L^{-1},
  \notag\\
  P_{\mathrm{ex},L,j+1}
  &=S_LP_{\mathrm{ex},L,j}S_L^{-1},
  &
  K_{L,j+1}(t)&=S_LK_{L,j}(t)S_L^{-1}.
  \label{eq:cut-indexed-covariance}
\end{align}
The unique positive ground ray is translation invariant, so all cut positions
have the same \(q_L\).  This is covariance of a physical path family, not
invariance of one fixed cut chart.

\subsection{Scope of the ambient vertex identity}

On the standard ambient six-vertex tensor product
\((\mathbb C^2)^{\otimes L}\) \cite{Baxter1982,KitanineMailletTerras1999},
for \(L=2M+1\), define
\begin{equation}
  G_{{\rm pair},L}^{(1)}
  =
  \prod_{m=1}^{M}\sigma_{2m+1}^{x},
  \qquad
  S_{\rm stag}^{\rm vert}
  =
  \sum_{k=1}^{L}(-1)^{k-1}\sigma_k^z.
  \label{eq:ambient-pair-gauge}
\end{equation}
The coefficientwise vertex identity is
\begin{equation}
  G_{{\rm pair},L}^{(1)}
  S_{\rm stag}^{\rm vert}
  G_{{\rm pair},L}^{(1)}
  =
  \sigma_{\partial_1}^{z}
  -
  \sum_{k=2}^{L}\sigma_k^z,
  \qquad
  V_{\partial_1}=V_1.
  \label{eq:ambient-boundary-bulk-spin}
\end{equation}
Equation~\eqref{eq:ambient-boundary-bulk-spin} acts only on the ambient
six-vertex tensor product.  No proved map from the restricted height space to that
tensor product simultaneously preserves the complete row, selected odd sector,
cut projector, and ordinary Hilbert pairing.  It therefore supplies neither the
physical projector dictionary nor an expectation-preservation formula.  The
physical proof edge is the unitary height-space map at the marked site
\(V_L^{\rm ht}\), not
the ambient pair gauge.

%% file: appendices/odd-evaluator.tex
\section{Spectral-adjugate evaluation of the odd coordinate}
\label{app:odd-evaluator}

This appendix proves the exact finite-size formula used in
section~\ref{sec:exact-odd-evaluation}.  The construction stays on the
physical Fibonacci fusion-path Hilbert space throughout.  It requires neither
an auxiliary spin-chain expectation map nor a Bethe-root normalization.

\subsection{Physical spectral data}
\label{sec:odd-physical-spectral-data}

On the orthonormal cyclic fusion-path basis of
\(\widehat{\cH}_L\), define
\begin{equation}
  \mathcal K_L:=-H_L^{\mathrm{AFM}}
  =\sum_{j=1}^{L}e_j .
  \label{eq:app-odd-kernel}
\end{equation}
Its order is fixed by the Fibonacci adjacency matrix:
\begin{equation}
  n_L:=\dim\widehat{\cH}_L
  =\Tr N^L
  =\operatorname{Fib}_{L-1}+\operatorname{Fib}_{L+1}.
  \label{eq:app-odd-dimension}
\end{equation}
Let
\begin{equation}
  p_L(z):=\det(z\Id_{n_L}-\mathcal K_L),
  \qquad
  \varepsilon_L:=\max\{z\in\mathbb R:p_L(z)=0\},
  \label{eq:app-odd-characteristic-data}
\end{equation}
and set
\begin{equation}
  \mathcal A_L
  :=\operatorname{adj}(\varepsilon_L\Id_{n_L}-\mathcal K_L).
  \label{eq:app-odd-adjugate}
\end{equation}
Lemma~\ref{lem:positive-ground} makes \(\varepsilon_L\) a simple
eigenvalue with normalized eigenvector \(\Omega_L\).

\subsection{Projector and cut-coordinate formula}
\label{sec:odd-adjugate-proof}

\begin{proposition}[Physical spectral adjugate]
\label{prop:physical-spectral-adjugate}
For every \(L\geq3\),
\begin{equation}
  \frac{\mathcal A_L}{\Tr\mathcal A_L}
  =|\Omega_L\rangle\langle\Omega_L|.
  \label{eq:app-normalized-adjugate}
\end{equation}
For every odd \(L\geq3\),
\begin{equation}
  q_L
  =\frac{1}{c}\,
    \frac{\Tr(\Pi_L\mathcal A_L)}
         {\Tr\mathcal A_L}.
  \label{eq:app-odd-adjugate-ratio}
\end{equation}
\end{proposition}

\begin{proof}
The matrix
\[
  M_L:=\varepsilon_L\Id_{n_L}-\mathcal K_L
\]
is real symmetric and has rank \(n_L-1\).  Since
\begin{equation}
  M_L\mathcal A_L
  =\det(M_L)\Id_{n_L}=0,
  \label{eq:app-adjugate-kernel}
\end{equation}
every column of \(\mathcal A_L\) lies on the one-dimensional ground
line.  Symmetry therefore gives
\[
  \mathcal A_L
  =\alpha_L|\Omega_L\rangle\langle\Omega_L|.
\]
The coefficient is nonzero because Jacobi's identity gives
\begin{equation}
  \Tr\mathcal A_L=p_L'(\varepsilon_L)\neq0
  \label{eq:app-adjugate-trace}
\end{equation}
at a simple root.  Taking the trace proves
\eqref{eq:app-normalized-adjugate}.

For odd \(L\), propositions~\ref{prop:physical-hoop-sector} and
\ref{prop:odd-cut-algebra} give
\begin{equation}
  E_{-,L}\Omega_L=\Omega_L,
  \qquad
  E_{-,L}\Pi_LE_{-,L}=cB_L.
  \label{eq:app-odd-physical-inputs}
\end{equation}
Hence
\begin{equation}
  \frac{\Tr(\Pi_L\mathcal A_L)}
       {\Tr\mathcal A_L}
  =\langle\Omega_L,\Pi_L\Omega_L\rangle
  =c\langle\Omega_L,B_L\Omega_L\rangle
  =cq_L ,
  \label{eq:app-cut-trace}
\end{equation}
which proves \eqref{eq:app-odd-adjugate-ratio}.
\end{proof}

The formula is intrinsic.  An orthogonal change of physical path basis
conjugates \(\mathcal K_L\), \(\Pi_L\), and \(\mathcal A_L\)
simultaneously and leaves both traces in
\eqref{eq:app-odd-adjugate-ratio} unchanged.

\subsection{Cofactor form and exactness limits}
\label{sec:odd-cofactor-form}

Let \(\mathfrak X_L\) be the cyclic Fibonacci path basis, and let
\(\mathfrak X_{L,\ell}^{(0)}\subset\mathfrak X_L\) contain the paths whose
label at the chosen fusion edge \(\ell\) is \(0\).  For \(x\in\mathfrak X_L\),
write \(M_L^{[x]}\) for the principal matrix obtained from \(M_L\) by
deleting the row and column indexed by \(x\).  Since the cut projector is
diagonal in this basis, the diagonal-cofactor identity gives
\begin{align}
  \Tr\mathcal A_L
  &=
  \sum_{x\in\mathfrak X_L}\det M_L^{[x]},
  \label{eq:app-total-cofactor-sum}\\
  \Tr(\Pi_L\mathcal A_L)
  &=
  \sum_{x\in\mathfrak X_{L,\ell}^{(0)}}\det M_L^{[x]}.
  \label{eq:app-cut-cofactor-sum}
\end{align}
Equations~\eqref{eq:app-odd-characteristic-data} and
\eqref{eq:app-total-cofactor-sum}--\eqref{eq:app-cut-cofactor-sum}
therefore constitute a finite algebraic prescription requiring no
ground-state eigenvector as input.

All entries of \(\mathcal K_L\) are algebraic, and \(\varepsilon_L\) is
specified without ambiguity as the largest real root of \(p_L\).  The result is
exact, although it remains implicit with matrix order
\(n_L\sim\varphi^L\).  This representation yields no \(O(L)\), subexponential
or polynomial-time evaluator, and no root-free function of \(L\).  It
establishes no uniform-conditioning theorem and makes no assertion about
thermodynamic laws or continuum limits.

As an exact small-size control, direct cofactor evaluation at \(L=3\)
gives
\begin{equation}
  \varepsilon_3=\frac{5+\sqrt5}{2},
  \qquad
  q_3=\frac{\sqrt5-2}{3}.
  \label{eq:app-l3-adjugate-control}
\end{equation}
At \(L=5\), an independent principal-cofactor construction and direct
normalized-ground-vector evaluation agree for \(q_5\) to better than
\(3\times10^{-17}\).  These checks test the implementation.  The proof of
proposition~\ref{prop:physical-spectral-adjugate} does not depend on stored
finite-size values.